\documentclass[twocolumn,showpacs,floatfix,nofootinbib,showkeys,superscriptaddress]{revtex4-1}
\usepackage[T1]{fontenc}
\usepackage[latin9]{inputenc}
\usepackage{amssymb,amsmath,amsthm,amsfonts}
\usepackage{graphicx}
\usepackage{esint}
\usepackage{babel}
\usepackage{color}
\usepackage{enumitem}
\usepackage{subfigure}
\usepackage{comment}
\usepackage[colorinlistoftodos,prependcaption,textsize=tiny]{todonotes}
\usepackage{cancel}

\usepackage{hyperref}
\allowdisplaybreaks
\hypersetup{
colorlinks=true,
linkcolor=blue,
linktocpage=true,
citecolor=violet
}

\newtheorem{theorem}{Theorem}[section]
\newtheorem{lemma}{Lemma}[section]
\newtheorem{prop}{Proposition}[section]

\newtheorem{cor}{Corollary}[section]
\newtheorem{con}{Conjecture}

\newcommand{\munich}{Max-Planck-Institut f\"ur Physik, Werner-Heisenberg-Institut, 
F\"ohringer Ring 6, D-80805 M\"unchen, Germany}

\def\Mathematica{{{\sc Mathematica}}}
\def\Dlog{{{\sc DlogBasis}}}

\newcommand{\liverpool}{Department of Mathematical Sciences, University of Liverpool, Liverpool L69 3BX, 
U.K.
}

\begin{document}
\preprint{MPP-2022-129}

\title{Landau and leading singularities in arbitrary space-time dimensions}

\author{Wojciech Flieger}
\affiliation{Dipartimento di Fisica e Astronomia, Universit\`a degli Studi di Padova, Via Marzolo 8, I-35131 Padova, Italy}
\affiliation{INFN, Sezione di Padova, Via Marzolo 8, I-35131 Padova, Italy}
\affiliation{\munich}
\author{William J. Torres~Bobadilla}
\affiliation{\munich}
\affiliation{\liverpool}

\begin{abstract}
Using the decomposition of the $D$-dimensional space-time into parallel and perpendicular subspaces, we study and prove a connection between Landau and leading singularities for $N$-point one-loop Feynman integrals by applying the multi-dimensional theory of residues. We show that if $D=N$ and $D=N+1$, the leading singularity corresponds to the inverse of the square root of the leading Landau singularity of the first and second type, respectively. We make use of this outcome to systematically provide differential equations of Feynman integrals in canonical forms and the extension of the connection of these singularities at the multi-loop level by exploiting the loop-by-loop approach. Illustrative examples with the calculation of Landau and leading singularities are provided to supplement our results.  
\end{abstract}


\maketitle

\graphicspath{{figs/}}

\section{Introduction}
The study of Feynman integrals as analytic functions has a long history. It is already apparent in the definition of Feynman propagators, which is the integral in the complex plane with the poles moved from the real axis by the $+\imath 0$ prescription allowing integration over real momenta. However, the real rise of the analytic approach to Feynman integrals, or the scattering matrix in general, began at the turn of the 50s and 60s of the previous century with the works of Landau, Bjorken and Nakanishi, who introduced the notion of what is presently known as Landau equations and the Landau singularity~\cite{Bjorken_phd,Landau:1959fi, nakanishi}. 
Landau equations give a necessary condition for a function defined by an integral to have singularities. 
Together with an attempt to formulate nuclear interactions in terms of the scattering matrix in a program called S-matrix theory, there had been a period of intensive studies of the analytic properties of the scattering matrix and Feynman integrals. This program resulted in multifarious methods of studying the analytic structure of Feynman integrals including Hadamard's lemma, topological methods, differential equations, and series expansion, to mention some~\cite{eden_1966,hwa_1966,TODOROV,Golubeva_1976}. 

After the success of QCD, however, the S-matrix program had been mostly abandoned, leaving one branch which focused on differential equations for generalised functions still going for a few more years in terms of micro-local analysis and the hyper-function theory~\cite{Iagolnitzer:1969sk,cmp/1103904075}. In the 90s of the XX century there began a renaissance of analytic studies of scattering amplitudes and Feynman integrals, which started with works on unitarity and cuts~\cite{Bern:1996je,Bern:1997sc,Britto:2004nc}, some other important developments include introduction of amplituhedron whose definition is based on the singularity structure~\cite{Arkani-Hamed:2013jha,Arkani-Hamed:2013kca}. In this work, we focus on two fundamental concepts, Landau singularities and leading singularities. 
The first one results from a system of equations, and describes a variety in a space spanned by kinematic variables, which is a singular locus of integrals depending on parameters allowing for a study of their analytic structure. 
For mathematical definitions of Landau varieties, we refer the Reader to~\cite{Pham2011full,Mizera:2021icv}. Recently, Landau singularities again grabbed the attention of the community and interesting works appeared focusing on mathematical aspects of these structures~\cite{Brown:2009ta,Bloch:2010gk,Abreu:2017ptx,Collins:2020euz,Berghoff:2020bug,Muhlbauer:2020kut,Hannesdottir:2021kpd,Mizera:2021icv,Muhlbauer:2022ylo}.

The second notion describes the maximal residue of Feynman integrals \cite{Cachazo:2008vp}. The knowledge of leading singularities plays a crucial role in the analytic calculation of 
Feynman integrals through differential equations methods~\cite{Kotikov:1990kg,Remiddi:1997ny,Gehrmann:1999as}.
The identification of $d\log$ structures has led to 
a systematic procedure to evaluate Feynman integrals 
in terms of transcendental numbers and special functions~\cite{Henn:2013pwa}. 
Therefore, the connection between Landau and leading singularities needs to be provided and mathematically supported. 
This is, in effect, the main target of this communication, 
in which we will show that for any scalar one-loop Feynman integral the leading singularities can be expressed by the leading Landau singularities, of the first and second type, by accounting for the appropriated
space-time dimension. 

In order to elaborate on the connection between Landau and leading singularities,
we depth on the study of  
Feynman integrals in an arbitrary space-time dimension, by performing an analysis at the integrand level. 
This is carried out through a decomposition of the space-time dimension in terms of two independent and complementary subspaces: parallel and perpendicular. 
This decomposition has been initially considered in~\cite{Czarnecki:1994td,Frink:1996ya,Kreimer:1991wj,Kreimer:1992zv} 
and recently applied to simplifications of multi-loop scattering amplitudes
within the integrand reduction method framework~\cite{Mastrolia:2016dhn,Mastrolia:2016czu}. 
We analyse in detail the analytic and singular structure of scalar one-loop Feynman integrals with $N$ external momenta (often referred to as $N$-gons, and as they are interesting objects there has already been effort to understand their properties~\cite{Ellis:2007qk,Dixon:2011ng,DelDuca:2011ne,Papadopoulos:2014lla,Spradlin:2011wp,Kozlov:2015kol,Abreu:2017ptx,Bourjaily:2019exo}).\\

\noindent\textbf{Contributions}
\begin{itemize}
\item By elaborating on the decomposition of the space-time dimension into two independent subspaces (parallel and perpendicular) and an extensive use of the multivariate theory of residues given by Leray, we prove that for one-loop integrals leading singularities are expressed as the inverse of the square root of Landau singularities. 
We study the dimensional dependence of this relation by showing when the leading singularity corresponds to the leading Landau singularity of the first and second type. 
The formal part of the study is supplemented by many concrete examples.
\item Having at hand a method to extract leading singularities by simply looking at Landau singularities, we discuss the application of this relation in the analytic calculation of Feynman integrals by the method of differential equations~\cite{Kotikov:1990kg,Remiddi:1997ny}. We show that within our approach the so-called canonical form of the differential equation is straightforwardly obtainable regardless of the space-time dimension~\cite{Henn:2013pwa}. 
\item On top of the studies carried out for one-loop Feynman integrals, we employed the multivariate theory of residues given by Leray~\cite{Leray_1959} 
in the calculation of leading singularities of two-loop planar and non-planar three-point functions with massive external and massless internal momenta. 
As a by-product of this direct calculation, we elaborated on the loop-by-loop approach and providing the $L$-loop ladder diagrams for three and four-point functions.
Continuing on the multi-loop application of the method we discuss the canonical differential equation for heavy-quark pair production through parton-parton annihilation in QCD at the two-loop level.
\end{itemize}

\noindent\textbf{Outline}

\noindent This paper is organised as follows. In Section~\ref{sec:decomposition}, we discuss Feynman integrals in arbitrary space-time dimension $D$. Then, we introduce an alternative representation of Feynman integrals which splits space-time into parallel and perpendicular components, this will lay a basis for our studies. In Section~\ref{sec:landau}, we introduce the notion of Landau singularities, including simple examples. Then we derive an equation for the Landau singularity in the parallel and perpendicular representation, showing that it obeys a very simple form, and all the information about the singularity is included in the perpendicular component. This allows us for an alternative derivation of the Landau singularity based on on-shell conditions. Section~\ref{sec:leading} is devoted to leading singularities, where we give illustrative one-loop examples. 
Then, we discuss and formulate the two main theorems concerning the dependence of leading singularities on the space-time dimension and its connection to leading Landau singularities. This section ends with an application of the obtained results in the context of the method of differential equations for Feynman integrals. Section~\ref{sec:twoloop} is a prelude to the multi-loop investigation of Landau and leading singularities and the connection between them, supported by examples of integrals with leading singularities. The work ends with conclusions and a discussion of possible directions for further studies. The paper is aided by Appendices containing mathematical background on Leray's theory of residues (\ref{app_leray}), supplementary results which are used in the derivation of main results (\ref{app_supp}), more examples of one-loop leading singularities (\ref{app:landexamples}), and the proof of the main theorem in the momentum representation (\ref{app:alt_proof}).

\section{Feynman integrals in arbitrary space-time dimensions}
\label{sec:decomposition}

\begin{figure}[t]
\centering
\includegraphics[scale=0.9]{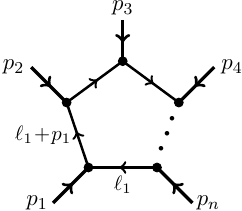}
\caption{Loop-momentum configuration for one-loop $n$-point Feynman integrals (N-gons).}
\label{fig:gen1L}
\end{figure}

Let us start by recalling that a $D$-dimensionally
regularised Feynman integral, where $D$ is a complex number, at $L$ loops with $N$ external momenta can be written as, 
\begin{align}
J_{N}^{\left(L\right),D}\left(1,\hdots,n;n+1,\hdots,m\right) & \nonumber\\
=\int\prod_{i=1}^{L}\frac{d^{D}\ell_{i}}{\imath\,\pi^{D/2}}
&
\frac{\prod_{k=n+1}^{m}D_{k}^{\nu_{k}}}{\prod_{j=1}^{n}D_{j}^{\nu_{j}}}\,, 
\label{eq:intL}
\end{align}
where the first $n$ propagators, with integer powers, $\nu_i\in\mathbb{Z}$, 
represent the loop topology and are expressed as  
\begin{align}
D_i =  q_{i}^{2} - m_i^2+\imath0\,.
\label{eq:genprop}
\end{align}
Here, $q_{i}$ contains the dependence on loop ($\ell_i$) and external ($p_i$) 
momenta.
For instance, in one-loop topologies, we can set, 
$q_{i}=\ell_1+P_{i-1}$, with $P_{i-1}=p_1+\hdots+p_{i-1}$ and $P_0=0$, 
for a combination of external momenta (see Fig.~\ref{fig:gen1L}).
For the $D$-dimensional integral~\eqref{eq:intL}, loop and external momenta are elements of infinite dimensional vector space. The mass of the internal propagators and infinitesimal Feynman prescription 
are, respectively, $m_i$ and $+\imath 0$.
Throughout this paper, we will consider different kinematic configurations, hence, 
no assumption on internal and external masses is given at this point. 

The remaining $(m-n)$ propagators in~\eqref{eq:intL}
correspond to irreducible scalar
products often named as auxiliary propagators,
whose structure can be chosen as in Eq.~\eqref{eq:genprop} or as
a product between external and internal momenta, $2\ell_i\cdot p_j$. 
For this set of propagators, their powers
are only positive integer numbers, $\nu_i\in\mathbb{Z}_{\geq 0}$.

\subsection{Feynman integrals in parallel and perpendicular space}
\label{sec:decosub}

Since Feynman integrals can be characterised by the number of external momenta,
vertices, and internal propagators, 
one can think of a decomposition
in terms of two independent and complementary subspaces,
which we will refer
to parallel, with dimension $D_{\parallel}$, and perpendicular, of dimension $D_{\perp}$ (we will use the same symbols for the dimensions and spaces itself, the meaning will be clear from the context),
\begin{align}
D &= D_\parallel + D_\perp\,,
\label{eq:decoD}
\end{align}
in which the first one, say $D_{\parallel}=\min\left(\left\lfloor D\right\rfloor ,E\right)$, 
is spanned by the independent external momenta, $E$.\footnote{Here 
$\left\lfloor D\right\rfloor $ corresponds to the integer and real part
of the space-time dimension.}
The second subspace, on the contrary, is spanned by momenta that are 
orthogonal to the external momenta, $D_\perp=D-D_{\parallel}$. 

Throughout this paper we follow decomposition~\eqref{eq:decoD} to cast 
all information of external momenta in $D_\parallel$.
We will observe in the following discussion that, regardless of the space-time dimension, it is possible to recast the formal $D$-dimensional integral \eqref{eq:intL} as an ordinary $m$-dimensional integral,
where $m$ corresponds to the number of propagators, 
showing in this way a clear connection with the Baikov representation~~\cite{Baikov:1996iu,Baikov:1996rk}. 
We emphasize that to carry out this analysis, we do not consider a particular dimension, 
instead, we perform this analysis by considering the most general features of Feynman integrals
in the context of leading and Landau singularities. 

Since we perform an analysis at the integrand level, one can translate the 
decomposition of the space-time dimension to the parametrization of the 
$D$-dimensional loop momenta, which, within this framework, we can express as,\footnote{We 
closely follow the notation of~\cite{Gnendiger:2017pys} to make evident
the dimension where loop momenta live in.} 
\begin{align}
\ell_{i,\left[D\right]}^{\alpha}&=\ell_{i,\left[D_{\parallel}\right]}^{\alpha}+\ell_{i,\left[D_{\perp}\right]}^{\alpha}\,,
\label{eq:loopparam}
\end{align}
that in the parametrization of parallel and perpendicular directions of the loop momenta, 
we have, by definition,
$\ell_{i,[D_{\parallel}]}^\alpha \ell_{j,[D_{\perp}]\,\alpha}\equiv
\ell_{i,[D_{\parallel}]}\cdot \ell_{j,[D_{\perp}]}=0$.

\subsection{One-loop decomposition}

The decomposition into parallel and perpendicular subspaces will help us to make sense of the formal $D$-dimensional integral~\eqref{eq:intL} and transform it to the form of finite (regularised) dimensional integrals which will be suitable for our analysis. We will focus only on scalar Feynman integrals and present arguments for the one-loop case. 
Multi-loop Feynman integrals within this representation
will be discussed in the next section.

The one-loop scalar Feynman integral in the discussed decomposition \eqref{eq:loopparam} can be written as,
\begin{align}
J_{N}^{\left(1\right),D}=\frac{1}{\imath\pi^{D/2}}\int\prod_{i=1}^{D_{\parallel}}d\ell_{\parallel}^{i}\int\,d^{D_{\perp}}\ell_{\perp}\frac{1}{\prod_{j=1}^{N}D_{j}}\,,
\end{align}
where to simplify notation, 
$\ell_{\parallel}=\ell_{1,\left[D_{\parallel}\right]}$
and $\ell_{\perp}=\ell_{1,\left[D_{\perp}\right]}$.

Since scalar Feynman integrals depend only on scalar products between 
external and loop
momenta, we can transform integration over perpendicular components to spherical coordinates and perform 
integration over angular components, 
\begin{align}
J^{\left(1\right),D}_N
&=
\frac{\Omega_{D_{\perp}-1}}{\imath \pi^{D/2}}
\int
\prod_{i=1}^{D_{\parallel}}d\ell_{\parallel}^{i} \int_{0}^{\infty}
\,d\ell_{\perp}\ell_{\perp}^{D-D_{\parallel}-1}\frac{1}{\prod_{j=1}^{N}D_{j}}
\,,
\label{eq:1lop}
\end{align}
In effect we obtained finite-dimensional integral of $D_{\parallel}+1$ variables $\ell_{\parallel}^{1},\hdots,\ell_{\parallel}^{D_{\parallel}},\ell_{\perp}$ which will serve as a definition of the formal $D$-dimensional integral \eqref{eq:intL}. The denominators $D_{j}$ are functions of the vectors $q_{j}^{\alpha} = q_{j, \parallel}^{\alpha} + q_{j, \perp}^{\alpha}$, which in view of the above definition can be defined as the $D_{\parallel}+1$ dimensional vectors 
(i.e., $\alpha=1,\hdots, D_{\parallel}+1$).
As we will be mostly concerned with the integrands, this form is enough for our analysis. In general, convergence of~\eqref{eq:1lop} depends on $D$, and the convergence region can be determined by analytic continuation in $D$~\cite{tHooft:1972tcz}. 

To elucidate further the convergence of integral~\eqref{eq:1lop}, let us consider the particular case
of $D=4-2\epsilon$ space-time dimensions in the four-dimensional helicity scheme~\cite{Bern:1991aq}.
Integration measure in one-loop Feynman integrals, as suggested by Ref.~\cite{Bern:1995db},
can be decomposed into two independent subspaces,
\begin{align}
d^D\ell_1 = d^4\ell_{1,[4]}\,d^{-2\epsilon}\ell_{1,[-2\epsilon]}\,.
\end{align}
This decomposition, after integration over the angular components of spherical coordinates in the additional $-2\epsilon$ dimension
and keeping in mind external momenta in $D=4$, 
allows to cast the wedge product in $D=4-2\epsilon$ space-time dimensions as, 
\begin{align}
d^{4-2\epsilon}\ell_1 =
d\ell_1^1\wedge\hdots\wedge d\ell_1^4 \wedge d\mu^2\,
\end{align}
where $\mu^2=-\ell_{1,[-2\epsilon]}^\alpha\ell_{1,[-2\epsilon]\,\alpha}$. 
With this decomposition in mind, we continue our analysis in parallel and perpendicular subspaces. 

The dimension of the parallel subspace is associated with the number of linearly independent external momenta. Let us then assume that this subspace is spanned by independent external momenta,
\begin{align}
\ell_{\parallel}^{\alpha} = 
\ell_{1,\left[D_{\parallel}\right]}^{\alpha} 
& =\sum_{j=1}^{D_{\parallel}}a_{j}\,p_{j}^{\alpha}\,.
\end{align}
We now introduce the following change of variables $\ell_{\perp}^2 = \lambda_{11}$, with the differential $d\ell_{\perp} = \frac{1}{2}\lambda_{11}^{-\frac{1}{2}}d\lambda_{11}$. 
In the new variables, integral~\eqref{eq:1lop} is given by,
\begin{align}
J^{\left(1\right),D}_N
&=
\frac{\mathcal{J}_{\left[D_{\parallel}\right]}^{(1)}\,
\Omega_{D_{\perp}-1}}{2 \imath \pi^{D/2}}
\int
\prod_{i=1}^{D_{\parallel}}da_{i}
\,d\lambda_{11}\,\lambda_{11}^{(D_\perp-2)/2}\frac{1}{\prod_{j=1}^{N}D_{j}}
\,,
\label{eq:int1L}
\end{align}
with $\mathcal{J}_{\left[D_{\parallel}\right]}$ being 
the Jacobian of the transformation of the parallel loop-momentum components ($a_i$). To find this Jacobian, we write the parallel loop momentum in terms of components,
\begin{align}
\ell_{\parallel}^{1}= & \sum_{j=1}^{D_{\parallel}}p_{j}^{1}a_{j}\,,\nonumber \\
\ell_{\parallel}^{2}= & \sum_{j=1}^{D_{\parallel}}p_{j}^{2}a_{j}\,,\nonumber \\
\vdots\nonumber \\
\ell_{\parallel}^{D_{\parallel}}= & \sum_{j=1}^{D_{\parallel}}p_{j}^{D_{\parallel}}a_{j}\,.
\end{align}
Now, we can express the differential form $d\ell_{\parallel}^{1}\wedge\hdots\wedge d\ell_{\parallel}^{D_{\parallel}}$ in the new variables ($a_i$),
\begin{align}
d\ell_{\parallel}^{1}\wedge\hdots d\ell_{\parallel}^{D_{\parallel}}=\left(\sum_{j=1}^{D_{\parallel}}p_{j}^{1}\,da_{j}\right)\wedge\hdots\wedge\left(\sum_{j=1}^{D_{\parallel}}p_{j}^{D_{\parallel}}\,da_{j}\right)\,,
\end{align}
that, because of Lemma~\ref{wedge_det}, amounts to, 
\begin{align}
d\ell_{\parallel}^{1}\wedge\hdots\wedge d\ell_{\parallel}^{D_{\parallel}}= & \det\left(J\right)da_{1}\wedge\hdots\wedge da_{D_{\parallel}} \\
= & \begin{vmatrix}p_{1}^{1} & \hdots & p_{D_{\parallel}}^{1}\\
p_{1}^{2} & \hdots & p_{D_{\parallel}}^{2}\\
\vdots & \ddots & \vdots\\
p_{1}^{D_{\parallel}} & \hdots & p_{D_{\parallel}}^{D_{\parallel}}
\end{vmatrix}da_{1}\wedge\hdots\wedge da_{D_{\parallel}}\,.
\notag
\end{align}
Notice that the determinant of the Jacobian matrix $J$ of this transformation in the above form is not very insightful. 
To put it in a more useful form, we make use of the identity 
$\det(J^{T}gJ)= \det(J)^{2}\det(g) \Rightarrow \det(J)=\pm \sqrt{\det(J^{T}gJ)/\det(g)}$, 
where $g=\text{diag}(1,-1,\hdots,-1)$ 
is the metric tensor in Minkowski space with, $\det(g)=(-1)^{D_{\parallel}-1}$.
By calculating the product, 
\begin{widetext}
\begin{equation}
J^{T}gJ=
\left(
\begin{array}{cccc}
p_{1}^{1} & p_{2}^{2} & \hdots & p_{1}^{D_{\parallel}} \\ 
\vdots  & \vdots & \ddots  & \vdots \\
p_{D_{\parallel}}^{1} & p_{D_{\parallel}}^{2} & \hdots & p_{D_{\parallel}}^{D_{\parallel}}
\end{array}  
\right)
\left(
\begin{array}{cccccc}
1 & 0 & \hdots & 0 & \hdots & 0 \\
0 & -1 & \hdots & 0 & \hdots & 0 \\
\vdots  & \vdots & \hdots & \vdots & \ddots & \vdots \\
0 & 0 & \hdots & 0 & \hdots & -1 \\
\end{array}  
\right)
\left(
\begin{array}{cccccc}
p_{1}^{1} & \hdots & p_{D_{\parallel}}^{1} \\
p_{1}^{2} & \hdots & p_{D_{\parallel}}^{2} \\
 \vdots & \ddots & \vdots \\
p_{1}^{D_{\parallel}} & \hdots & p_{D_{\parallel}}^{D_{\parallel}} \\
\end{array}  
\right)
\,,
\end{equation}
\end{widetext}
we get,
\begin{equation}
J^{T}gJ = 
\left(
\begin{array}{ccc}
p_{1} \cdot p_{1} & \hdots & p_{1} \cdot p_{D_{\parallel}} \\
\vdots & \ddots & \vdots \\
p_{1} \cdot p_{D_{\parallel}} & \hdots & p_{D_{\parallel}} \cdot p_{D_{\parallel}}
\end{array}
\right). 
\end{equation}
Therefore, we have,
\begin{equation}
\det(J^{T}gJ) = \det(p_{i}\cdot p_{j})\,.   
\end{equation}
Compiling all together, we have the following formula for the Jacobian,
\begin{align}
\mathcal{J}_{\left[D_{\parallel}\right]}=\det\left(J\right) & =\pm\sqrt{\left(-1\right)^{D_{\parallel}-1}\det\left(p_{i}\cdot p_{j}\right)}\nonumber \\
 & =\pm\left(\imath\right)^{D_{\parallel}}\sqrt{-\det\left(p_{i}\cdot p_{j}\right)}\,.
 \label{eq:jac1L}
\end{align}

After this preliminary observation on the integrand representation, let us comment on the structure of the integrand. 
Feynman propagators $D_j$'s,
due to parametrization~\eqref{eq:loopparam}, are
polynomials quadratic in  $a_{i}$ and linear in $\lambda_{11}$. 
For instance, 
in $D_\parallel=n-1>0$ and $D_\perp=D-n+1>0$,
Feynman propagators present in one-loop Feynman integrals, become, 
\begin{align}
D_{i}&=\left(\ell_{1}+P_{i}\right)^{2}-m_{i}^{2} 
\label{lambda_eq}\\
&=\sum_{j,k=1}^{n-1}
\left(p_{j}\cdot p_{k}\right)
\,a_{j}a_{k} \notag \\
&+2\sum_{j=1}^{n-1}\sum_{k=1}^{i-1}
\left(p_{j}\cdot p_{k}\right)
\,a_{j}  
+\sum_{j,k=1}^{i-1}
\left(p_{j}\cdot p_{k}\right)
-m_{i}^{2}+\lambda_{11}\,.
\notag
\end{align}

In the above decomposition, we notice that the analysis of any one-loop Feynman integral 
can be carried out by only keeping track of the number of independent external momenta. 
In effect, regardless of the space-time dimension, 
one needs to perform $D_\parallel+1$ (or $N$) integrations to elucidate the singular structure 
of the integral, since all integrations coming from the perpendicular subspace are reduced 
to a single one ($\lambda_{11}$).

\subsection{Multi-loop decomposition}
\label{sec:deco_multiloop}

Let us now continue our analysis with the generic multi-loop Feynman integral,
where inspired by the differential form obtained in the integration measure for one-loop Feynman integrals, 
one can think of an extension at $L$ loops given by, 
\begin{align}
\prod_{i=1}^{L}\frac{d^{D}\ell_{i}}{\imath\,\pi^{D/2}}&= 
d^{\frac{L\left(L+1\right)}{2}}\Lambda\,
d^{LD_{\parallel}}\ell_{\parallel}\,.
\label{eq:chanvars}
\end{align}
By taking into account the parametrization of loop momenta~\eqref{eq:loopparam}, 
relations~\eqref{eq:chanvars} take the explicit form, 
\begin{align}
 d^{LD_{\parallel}}\ell_{\parallel}&=\mathcal{J}_{\left[D_{\parallel}\right]}^{(L)}
 \prod_{i=1}^{L}\prod_{j=1}^{D_{\parallel}}da_{ij}\,,
\\
 d^{\frac{L\left(L+1\right)}{2}}\Lambda&=\frac{\Omega_{D_{\perp}}^{\left(L\right)}}{2^L}\,
 \left[G\left(\lambda_{ij}\right)\right]^{\frac{D_{\perp}-L-1}{2}}
 \prod_{1\leq i\leq j}^{L}d\lambda_{ij}\,,
\end{align}
with $G\left(\lambda_{ij}\right)$ corresponding to the Gram determinant of the 
perpendicular directions 
$\lambda_{ij}\equiv\ell_{i,[D_{\perp}]}\cdot\ell_{j,[D_{\perp}]}\equiv\ell_{i,[D_{\perp}]}^{\alpha}\ell_{j,[D_{\perp}]\,\alpha}$
and $\mathcal{J}_{\left[D_{\parallel}\right]}^{(L)}$ the Jacobian of the transformation
given by,
\begin{align}
\mathcal{J}_{\left[D_{\parallel}\right]}^{(L)}
&=\pm\left(\imath\right)^{LD_\parallel}\prod_{i=1}^{L}\sqrt{-\det\left(\frac{\partial\ell_{i,\left[D_{\parallel}\right]}^{\alpha}}{\partial a_{ij}}\frac{\partial\ell_{i,\left[D_{\parallel}\right]\,\alpha}}{\partial a_{ik}}\right)}\,,\\
&\text{with }j,k\leq D_{\parallel}\,, \notag
\end{align}
where the overall sign $\pm$ in the Jacobian originates from the square roots and, for the purpose of 
our discussion does not introduce any ambiguity. 
Likewise, notice that the overall prefactor of $\imath$ in the Jacobian does not play 
any important role. Therefore, in the discussion of the following sections, we will neglect it.

In view that our analysis is performed at the integrand level, we observe that all scalar products between
external and internal momenta (as well as propagators) are expressed in terms
of the loop-momentum components $a_{ij}$ and $\lambda_{ij}$.
A more involved structure that arises when considering helicity amplitudes
can methodically be integrated out through Gegenbauer polynomials, 
as elucidated in Ref.~\cite{Mastrolia:2016dhn}.

Therefore, similar to one-loop, the integration of spherical coordinates
in scalar multi-loop Feynman integrals are cast in, 
\begin{align}
\Omega_{D_{\perp}}^{\left(L\right)}&=\prod_{i=1}^{L}\frac{\Omega_{D_{\perp-i}}}{\pi^{D/2}}\,.
\end{align}

With the integration measure given by Eq.~\eqref{eq:chanvars}
in the expression of a generic multi-loop Feynman integral~\eqref{eq:intL}
at our disposal, 
an alert reader might notice a similarity with the Baikov representation. 
 In effect, this derivation has been performed in~\cite{Lee:2010wea,Frellesvig:2017aai,Primo:2017jtw}.
Within this framework, one obtains,
\begin{align}
d^{D_{\parallel}}\ell_{\parallel}\,d^{\frac{L\left(L+1\right)}{2}}\Lambda
\sim
\prod_{i=1}^{m}dD_i\,,
\end{align}
in which the number of integrations in l.h.s. and r.h.s. 
is the same,  $L(L+1)/2+EL=m$, as earlier anticipated. 

Let us emphasize that the decomposition in terms of parallel and perpendicular components has been initially employed to derive Baikov representation of Feynman integrals to understand dimensional recurrence relations between Feynman integrals~\cite{Lee:2010wea}, alternatively to their parametric representation in terms of graph polynomials~\cite{Tarasov:1996br,Lee:2009dh}.
In the following, we make an extensive use of the representation~\eqref{eq:chanvars} rather than Baikov representation because our main motivation lays in the mathematical study of the integration variables $a$ and $\lambda$.

Elaborating on this loop-momentum parametrization,  we can proceed 
to analyse Landau and leading singularities of different loop topologies.
We shall recall that the physical information is contained in the
parallel direction of the space-time dimension. The remaining integrations
are simply cast in $\lambda_{ij}$. 
Therefore, we will mostly draw our attention to integrations in 
the parallel subspace. 

\section{Landau Singularities}
\label{sec:landau}

We will focus on the form of the Landau equations in the decomposition of the space-time in terms of parallel and perpendicular components.

Let us first look at what form Landau equations take for our definition of $D$-dimensional integral~\eqref{eq:1lop}. 
Landau equations can be derived by requiring that polar sets $D_{j}$ are in non-general positions, i.e. when differential $dD_{j}$ are not linearly independent together with the requirement that we are confined to a singular variety of the integrand (on-shell conditions). 
Thus, we have, 
\begin{subequations}
\begin{align}
&D_{i} = 0 \, \text{ for } i=1,\hdots,N , \label{on-shell}\\
&\sum_{i}\alpha_{i}dD_{i}=  
\label{second_landau}
\notag \\
&\sum_{i} \alpha_{i}\left(\frac{\partial D_{i}}{\partial \ell_{\parallel}^{1}}d\ell_{\parallel}^{1}+\hdots+\frac{\partial D_{i}}{\partial \ell_{\parallel}^{D_{\parallel}}}d\ell_{\parallel}^{D_{\parallel}}
+\frac{\partial D_{i}}{\partial \ell_{\perp}}d\ell_{\perp} 
\right) = \notag \\
& \left (\sum_{i} \alpha_{i} \frac{\partial D_{i}}{\partial \ell_{\parallel}^{1}} \right)d\ell_{\parallel}^{1} + \hdots + \left (\sum_{i} \alpha_{i} \frac{\partial D_{i}}{\partial \ell_{\perp}} \right)d\ell_{\perp} = 0 \,.
\end{align}
\end{subequations}
Taking into account the independence of differential $d \ell$ the brackets in the last line of \eqref{second_landau} give the well-known form of the Landau equations.
In principle one can study singularities of subdiagrams by imposing $\alpha_{i}=0$ for some $i$, then Eq.~\eqref{on-shell} in general should be stated as $D_{i}=0 \text{ or } \alpha_{i}=0$. However, since we are interested in leading Landau singularities, all $\alpha$'s are non-zero. 

We impose that the Landau equations has nontrivial solution for $\alpha$'s which happens when $\det(\frac{\partial D_{i}}{\partial \ell^{j}})$ vanishes. By multiplying by the transpose of the matrix $(\frac{\partial D_{i}}{\partial \ell^{j}})$, this system can be put into the form,
\begin{equation}
\sum_{i} \alpha_{i} \left(q_{i}\cdot q_{j}\right)=0\,, \quad j=1,\hdots,N \,.
\label{eq:mysysLan}
\end{equation}
There is one such system for each independent loop in the diagram.
Together with on-shell conditions and momentum conservation, 
this gives us the equation for the leading Landau variety which we also call Landau singularity. In practice solving Landau equations for complicated Feynman diagrams can be very difficult. However, for one-loop diagrams, as we have only one equation of the type \eqref{eq:mysysLan} it simplifies a lot.
Thus, 
here and in the following, we define, for one-loop diagrams,
\begin{equation}
\text{LanS}_{n}^{(1)}  = \det(q_{i}\cdot q_{j})\,, \quad i,j=1,\hdots , N \,,
\label{one_loop_landau}
\end{equation}
which should be understood that on-shell conditions and momentum conservation are taken into account. Then, the solution $\text{LanS}_{n}^{(1)} =0$ gives us the Landau singularity for one-loop diagrams. 
This is a well established result know since the beginning of the analytic studies of Feynman integrals~\cite{eden_1966}.
In this respect, Eq.~\eqref{one_loop_landau} is given by the Gram determinant of internal momenta which can vanish identically if the number of internal momenta is larger than the dimension of the space. We will study the dimensional dependence of Landau singularities for one-loop diagrams. 

In the following, we will explore Landau equations in the parametrization of the loop momenta according to~\eqref{eq:loopparam}.
To start, we will observe in one-loop Feynman integrals that solutions to these 
equations, because of the decomposition of the dimension, 
are cast in $\lambda_{11}=\ell_\perp^2$ once on-shell conditions are imposed. 

\subsection{One-loop Landau singularities}
\label{sec:landau1L}

Because of the decomposition of the space-time dimension into parallel and perpendicular 
subspaces that was laid out in Sec.~\ref{sec:decosub}, one can work out, within this framework, the general solution of Eq.~\eqref{eq:mysysLan}. We will show that this provides the following expression for Landau singularities of a one-loop $n$-point Feynman integral, 
\begin{equation}
\text{LanS}_{n}^{(1)} 
= \lambda_{11} \det \left( (p_{i}\cdot p_{j})_{(n-1)\times (n-1)}\right),    
\end{equation}
which holds only when on-shell conditions are imposed.
To achieve that, let us first express the Gram determinant of propagators~\eqref{eq:mysysLan} 
in terms of parallel and perpendicular components,
\begin{prop}
\begin{equation}
\begin{split}
&{\rm{LanS}}_{n}^{(1)}  
= \det(q_{i} \cdot q_{j}) = \det((q_{i,\parallel} \cdot q_{j,\parallel} + \lambda_{11})_{n \times n})
\\
&= \det((q_{i,\parallel} \cdot q_{j,\parallel})_{n \times n}) \\
&+ \lambda_{11}\sum_{k=1}^{n}\det\left[ (q_{i,\parallel} \cdot q_{j,\parallel}(1-\delta_{jk}(1- \frac{1}{q_{i,\parallel} \cdot q_{j,\parallel}})))_{n \times n} \right]
\\
&= \lambda_{11}\sum_{k=1}^{n}\det\left[ (q_{i,\parallel} \cdot q_{j,\parallel}(1-\delta_{jk}(1- \frac{1}{q_{i,\parallel} \cdot q_{j,\parallel}})))_{n \times n} \right]\,,
\end{split}
\end{equation}
with $q_{i,\parallel}^\alpha = \ell_{1,[D_\parallel]}^\alpha+P_i^\alpha$, 
according to Fig.~\ref{fig:gen1L}.
\label{lan_perp}
\end{prop}
\begin{proof}
We work in a matrix form and write $q_{i,\parallel}\cdot q_{j,\parallel}$ as $q_{ij}$;
we compactly express $\det(q_{ij} + \lambda_{11})$ in a column form as follows,
\begin{equation}
\vert Q_{1} + \Lambda_{11}, Q_{2} + \Lambda_{11}, \hdots, Q_{n} + \Lambda_{11} \vert \,,
\end{equation}
where $Q_{j} = (q_{1j},\hdots,q_{nj})^{T}$ and $\Lambda_{11}$ should be understood as a column vector with all entries equal to $\lambda_{11}$. 

We now use multi-linearity of the determinant, starting from the additivity, 
\begin{equation}
\begin{split}
&\vert Q_{1}+\Lambda_{11},Q_{2}+\Lambda_{11},\hdots,Q_{n}+\Lambda_{11}\vert \\
&=\vert Q_{1},Q_{2}+\Lambda_{11},\hdots,Q_{n}+\Lambda_{11}\vert \\
&+\vert\Lambda_{11},Q_{2}+\Lambda_{11},\hdots,Q_{n}+\Lambda_{11}\vert\,.
\end{split}
\end{equation}
After repeating this procedure until having no more summation 
within any of the determinants and using the fact that if two columns of the determinant are the same it vanishes,
we are left with,
\begin{equation}
\begin{split}
&\vert Q_{1} + \Lambda_{11}, Q_{2} + \Lambda_{11},\hdots, Q_{n} + \Lambda_{11} \vert \\
&=\vert Q_{1} , Q_{2} ,\hdots, Q_{n} \vert + 
\vert \Lambda_{11}, Q_{2} ,\hdots, Q_{n}  \vert \\
&+\vert Q_{1}, \Lambda_{11},\hdots, Q_{n} \vert 
+ \hdots + 
\vert Q_{1} , Q_{2},\hdots, \Lambda_{11} \vert \,.
\end{split}
\end{equation}
Further, by applying the homogeneity, we get,
\begin{align}
\label{lambda_factor}
&\vert Q_{1} + \Lambda_{11}, Q_{2} + \Lambda_{11},\hdots, Q_{n} + \Lambda_{11} \vert  \notag
\\
&=\vert Q_{1} , Q_{2} ,\hdots, Q_{n} \vert + 
\lambda_{11}\vert I, Q_{2} ,\hdots, Q_{n}  \vert  \\
&+\lambda_{11}\vert Q_{1}, I ,\hdots, Q_{n} \vert
+ \hdots + 
\lambda_{11}\vert Q_{1} , Q_{2},\hdots, I \vert \,, 
\notag
\end{align}
where $I=(1,\hdots,1)^{T}$.

Moreover, as there are only $n-1$ linearly independent vectors $q_{i,\parallel}$, 
the $n \times n$ Gram determinant $\det((q_{i,\parallel} \cdot q_{j,\parallel})_{n \times n}) = \vert Q_{1} , Q_{2} ,\hdots, Q_{n} \vert$ vanishes. 
Thus, by factoring out $\lambda_{11}$ in~\eqref{lambda_factor} we obtain the assertion.
\end{proof}

Now we are in the position to show that the sum in Proposition \ref{lan_perp} is equal to the Gram determinant of independent external momenta.

\begin{prop}
$\sum_{k=1}^{n} \det(q_{i,\parallel} \cdot q_{j,\parallel}(1-\delta_{jk}(1-\frac{1}{q_{i,\parallel} \cdot q_{j,\parallel}})))$ is equal to $\det \left( (p_{i}\cdot p_{j})_{(n-1)\times (n-1)}\right)$. 
\label{prop_lambda}
\end{prop}
\begin{proof}
By making use of Lemma~\ref{sum_det}, we can write $\sum_{k=1}^{n} \det(q_{i,\parallel} \cdot q_{j,\parallel}(1-\delta_{jk}(1-\frac{1}{q_{i,\parallel} \cdot q_{j,\parallel}})))$ as $ \vert I (Q_{2}-Q_{1}) (Q_{3}-Q_{1})\hdots(Q_{n}-Q_{1}) \vert$. Let us subtract the last row from all other rows. This gives us $ \vert I' (Q'_{2}-Q'_{1}) (Q'_{3}-Q'_{1})\hdots(Q'_{n}-Q'_{1}) \vert $, 
where, in particular, $I'=(0,0,\hdots,0,1)^{T}$. 
By the Laplace expansion along the first column, we get the 
$(n-1)\times (n-1)$ determinant $(-1)^{n+1} \vert (Q''_{2}-Q''_{1}) (Q''_{3}-Q''_{1})\hdots(Q''_{n}-Q''_{1}) \vert $. Let us look at the elements of this determinant. The elements of $Q_{j}-Q_{1}$, for $j=2,\hdots,n$, are given by $q_{ij}-q_{i1}$. 
Then, by subtracting the last row, which has elements $q_{nj}-q_{n1}$, we get elements of $Q''_{j}-Q''_{1}$ given by $q_{ij}-q_{i1}-(q_{nj}-q_{n1}) = q_{ij}+q_{n1}-q_{i1}-q_{nj}$, for $i=1,\hdots, n-1$, but each $q_{ij} = \frac{1}{2}(q_{ii}+q_{jj} - f_{ij}(P))$ and thus all $q_{kk}$ cancel in each element, 
leaving only functions depending on the external momenta.

The function $f_{ij}(P)$ for each $q_{ij}$ has the following structure,
\begin{align}
f_{ij}(P)=\left(\sum_{k=i}^{j-1}p_{k}\right)^{2} \equiv f_{ij}\,,
\end{align}
where we assumed $i \leq j$, which is justified by the fact that $q_{ij}=q_{ji}.$\footnote{Notice 
that $f_{ij}(P)$ can also be written in a more symmetric way 
$f_{ij}(P)=\left(\sum_{k=1}^{\vert i-j\vert}p_{k-1+\frac{1}{2}(i+j-\vert i-j\vert)}\right)^{2}$, 
where we avoid this assumption.} 
Thus, we get,
\begin{equation}
\begin{split}
q_{ij}+q_{n1}-q_{i1}-q_{nj} = \frac{1}{2}( -f_{ij}+f_{n1}-f_{i1}+f_{nj}) \,.   
\end{split}
\end{equation}
Now, let us subtract $(i+1)$-th row from $i$-th row for each $i=1,\hdots,n-2$ to get,
\begin{equation}
\frac{1}{2} (-f_{ij} + f_{(i+1)j} + f_{i1} - f_{(i+1)1}) \,.     
\end{equation}
Next, we subtract $j$-th column from $(j+1)$-th column, for $j=2,\hdots,n$, which gives,
\begin{equation}
\frac{1}{2}(-f_{i(j+1)}+f_{ij}-f_{(i+1)j}+f_{(i+1)(j+1)}) \,.  
\end{equation}
Let us make the following change of variables $p_{i}=x_{i}-x_{i+1}$ (dual variables). Then, $f_{ij}(P)=(x_{i}-x_{j})^{2}$ and we have,
\begin{align}
&\frac{1}{2}\big(-(x_{i}-x_{j+1})^2+(x_{i}-x_{j})^2- \notag \\
&-(x_{i+1}-x_{j})^2+(x_{i+1}-x_{j+1})^2\big) =  \notag \\
&=(x_{i}\cdot x_{j+1}-x_{i} \cdot x_{j}+x_{i+1} \cdot x_{j}-x_{i+1}\cdot x_{j+1})= \notag \\
& = -(x_{i}-x_{i+1})(x_{j}-x_{j+1}) = -p_{i} \cdot p_{j}
\,.   
\end{align}%
Thus, by factoring out $-1$ from each column, we get $(-1)^{n-1}$,
and combining it with factors coming from the Laplace expansion, 
we finally obtain $(-1)^{n+1}(-1)^{n-1}=(-1)^{2n}=1$, which assures the statement.
\end{proof}

From the combination of Propositions~\ref{lan_perp} and~\ref{prop_lambda}, we support our claim, 
\begin{equation}
\lambda_{11} = \frac{{\rm{LanS}}_{n}^{(1)}}{\det \left( (p_{i}\cdot p_{j})_{(n-1)\times (n-1)}\right)}\,.
\label{lambda_landau}
\end{equation}
Let us emphasize that this relation is true only when on-shell conditions are imposed, as it is the way we defined $\text{LanS}_{n}^{(1)}$ in \eqref{one_loop_landau}. In other words, one can use on-shell conditions to eliminate parallel loop components and express $\lambda_{11}$ solely in terms of external kinematics, which we present below.

Notice that, due to the way how Feynman propagators are expressed 
in this parametrization (see Eq.~\eqref{lambda_eq}), one can determine $\lambda_{11}$
from the remaining Landau equations (i.e., on-shell conditions),
\begin{equation}
\begin{split}
D_{i} - D_{i+1} = 0\,, \\
D_{1} = 0\,.
\end{split}
\end{equation}
The first $n-1$ on-shell conditions $D_{i} - D_{i+1} = 0$ give $n-1$ equations for the loop-momentum
components $a_{i}$, for $i=1,\hdots n-1$, of Eq.~\eqref{lambda_eq}, 
which can then be expressed in terms of external kinematics and internal masses.
\begin{align}
0& = D_{i}-D_{i+1}\,,
\label{a_var}\\
0& = -m_{i}^{2}+m_{i+1}^{2}-2\sum_{j=1}^{i-1}\left(p_{i}\cdot p_{j}\right)-p_{i}^{2}-2\sum_{j=1}^{D_{\parallel}}\left(p_{i}\cdot p_{j}\right)\,a_{j}\,.
\notag
\end{align}
Furthermore, without loss of generality, we can consider the last equation $D_{1}=0$ (see Eq.~\eqref{lambda_eq})
to solve for $\lambda_{11}$,
whose explicit form at one-loop is given by, 
\begin{equation}
\lambda_{11} = m_{1}^{2} - \sum_{j,k=1}^{D_\parallel}
\left( p_{j}\cdot p_{k}\right)\, a_{j}a_{k}\,,
\label{lambda_eq1}
\end{equation}
where, by using the solutions~\eqref{a_var}, we can further express $\lambda_{11}$ 
as a function depending only on external kinematics and masses. 
From~\eqref{lambda_landau}, we also observe that $\lambda_{11}=0$ 
exactly corresponds to an equation for Landau variety.
\begin{cor}
For one-loop scalar Feynman integrals with 
equal internal masses, $\rm{LanS}$ can be split into 
a term proportional to the mass and a term independent of the mass. 
The proportionality factor for the mass term is the Gram determinant of independent external momenta, $\det(p_{i} \cdot p_{j})$.
\label{cor:corequalmass}
\end{cor}
\begin{proof}
It follows from \eqref{lambda_landau},~\eqref{a_var}, and~\eqref{lambda_eq1}.
\end{proof}

\subsection{One-loop bubble and triangle Landau singularities}

We finalise this section with examples of explicit 
computations of Landau singularities for one-loop scalar bubble and triangle integrals. 
For additional examples of up-to six-point Feynman integrals, we refer the Reader to 
Appendix~\ref{app:landexamples}. 

\paragraph{One-loop bubble integral\\}

For the one-loop scalar bubble the parametrization of the loop momentum,
by the decomposition studied in the paper is the following
\begin{align}
\ell_1^\alpha = a_1\,p_1^\alpha + \lambda_1^\alpha\,,
\end{align}
with $p_1^2\ne 0$. 
The Gram determinant of propagators, before imposing the on-shell conditions, becomes, 
\begin{align}
\det\left(q_{i}\cdot q_{j}\right)&=\lambda_{11}\,p_{1}^{2}\,,
\end{align}
then, from Eqs.~\eqref{a_var} and~\eqref{lambda_eq1} we obtain,
\begin{subequations}
\begin{align}
a_{1}&=\frac{1}{2p_{1}^{2}}\left(-m_{1}^{2}+m_{2}^{2}-p_{1}^{2}\right)\,,
\\
\lambda_{11}&=m_{1}^{2}-p_{1}^{2}\,a_{1}^{2}\,,
\end{align}
\end{subequations}
that corresponds to imposing on-shell conditions and allows to find the 
well-known Landau singularity, 
\begin{align}
\text{LanS}^{(1)}_2
&=-\frac{1}{4}\lambda_{\text{K}}\left(p_{1}^{2},m_{1}^{2},m_{2}^{2}\right)\,,
\end{align}
in terms of the K\"all\'en function, $\lambda_\text{K}(a,b,c)=a^2+b^2+c^2-2ab-2ac-2bc$. 

Furthermore, by considering the equal mass case, $m_{1}^{2}=m_{2}^{2}=m^{2}$, 
\begin{align}
\text{LanS}^{(1)}_2\Big|_{m_i^2 = m^2}
&
=-\frac{1}{4}\left(p_{1}^{2}\right)^{2}+m^{2}p_{1}^{2}\,.
\end{align}

\paragraph{One-loop triangle integral\\}

We continue with the explicit calculation of the Landau singularity of the one-loop triangle
and different kinematic scales.
The loop momentum can be parametrized in terms of two independent momenta,
\begin{align}
\ell_1^\alpha = a_1\,p_1^\alpha + a_2\,p_2^\alpha+ \lambda_1^\alpha\,,
\end{align}
with $p_1^2,p_2^2,p_3^2\ne 0$ and, because of the momentum conservation, 
$2p_1\cdot p_2 = p_3^2-p_1^2-p_2^2$. 
Thus, the loop-momentum components for this loop topology take the form, 
\begin{subequations}
\begin{align}
a_{1}&= \frac{1}{\lambda_{\text{K}}\left(p_{1}^{2},p_{2}^{2},p_{3}^{2}\right)}\bigg[2p_{2}^{2}\left(m_{1}^{2}-m_{2}^{2}+p_{1}^{2}\right) \notag \\
&-\left(p_{1}^{2}+p_{2}^{2}-p_{3}^{2}\right)\left(-m_{2}^{2}+m_{3}^{2}+p_{1}^{2}-p_{3}^{2}\right) \bigg]\,,
\\
a_{2}&= \frac{1}{\lambda_{\text{K}}\left(p_{1}^{2},p_{2}^{2},p_{3}^{2}\right)} \bigg[p_{1}^{2}\left(m_{1}^{2}+m_{2}^{2}-2m_{3}^{2}+p_{2}^{2}+p_{3}^{2}\right)  
\notag \\
&+\left(m_{1}^{2}-m_{2}^{2}\right)\left(p_{2}^{2}-p_{3}^{2}\right)-(p_{1}^{2})^{2}\bigg]\,,
\\
\lambda_{11}&=m_{1}^{2}-a_{1}^{2}p_{1}^{2}-a_{2}^{2}p_{2}^{2}+a_{1}a_{2}\left(p_{1}^{2}+p_{2}^{2}-p_{3}^{2}\right)\,,
\end{align}
\end{subequations}
whose combination amounts to the Landau singularity,
\begin{align}
\text{LanS}^{(1)}_3 &=
-\frac{1}{4}\Big[\frac{p_{1}^{2}\,p_{2}^{2}\,p_{3}^{2}}{3}-p_{3}^{2}\big(\left(m_{1}^{2}-m_{2}^{2}\right)\left(m_{2}^{2}-m_{3}^{2}\right) \notag \\
&+m_{2}^{2}\left(p_{1}^{2}+p_{2}^{2}-p_{3}^{2}\right)\big)\Big]+\text{cycl. perm.}\,.
\label{eq:land3pt}
\end{align}
The Landau singularity for the equal-mass case takes the form, 
\begin{align}
\text{LanS}^{(1)}_3\Big|_{m_i^2 = m^2}
&=-\frac{1}{4}\left(p_{1}^{2}\,p_{2}^{2}\,p_{3}^{2}+m^{2}\lambda_{\text{K}}\left(p_{1}^{2},p_{2}^{2},p_{3}^{2}\right)\right)\,.
\end{align}

\par\bigskip
Notice that in the equal-mass case in the 
above examples, the term proportional to $m^2$ corresponds to the Gram determinant of
the external momenta, in full agreement with corollary~\ref{cor:corequalmass}. 
For instance, in the one-loop triangle, we have,  
$\det\left(p_i\cdot p_j\right) = -1/4\,\lambda_{\text{K}}\left(p_{1}^{2},p_{2}^{2},p_{3}^{2}\right)$.

\section{Leading Singularities}
\label{sec:leading}

By continuing with the same strategy in the decomposition
of space-time dimension presented in Sec.~\ref{sec:decosub}, 
we now focus on leading singularities of one-loop Feynman integrals. 
By leading singularities we understand residues in Leray's sense evaluated over all polar sets, i.e.,
\begin{equation}
LS= \int_{\delta^{m} \sigma} \phi = (2 \pi i)^{m} \int_{\sigma} res^{m}[\phi]\,,    
\end{equation}
where $\phi$ is a closed form with $m$ polar sets $S_{1},\hdots,S_{m}$, $\sigma \in Z_{p-m}(S_{1} \cap \hdots \cap S_{m})$ is a cycle and $\delta^{m}\sigma$ is its composed co-boundary (see Appendix \ref{app_leray}). It can happen that the number of polar sets is larger than the dimension of the space, in this situation Leinartas decomposition \cite{leinartas1978, Aizenberg1994} allows to write $\phi$ as a sum of forms, each having at most as many polar sets as the dimension of the space.

To carry out this analysis, we present in this section illustrative computations 
of selected one-loop topologies, and provide the Reader with 
additional examples in Appendix~\ref{app:landexamples}. 

Additionally, because of the pattern displayed by explicit calculations
of leading singularities in the various examples,
we conjecture and prove the structure of these singularities in particular space-time dimensions -- a connection between Landau and leading singularities. 
Similar results to those presented in this section can already be found in~\cite{Abreu:2017ptx} which were obtained by working in the embedding formalism. 
Whilst authors of the mentioned paper focused on dimensional regularisation and present explicitly only analysis and results for even dimensions of the space-time, the odd dimensional cases can also by obtained by the method of~\cite{Abreu:2017ptx} by setting $\epsilon = -\frac{1}{2}$ 
\footnote{The closely related notion of maximal cut integrals presented in ~\cite{Abreu:2017ptx} is valid in arbitrary $D$ dimensions.}
.

In the last part of this section, we consider differential equations of one-loop Feynman integrals, 
showing that from the knowledge of their leading singularities 
the construction of differential equations in canonical form
is straightforwardly carried out.
We illustrate this with explicit calculations.

\subsection{One-loop leading singularities}

In the spirit of showing the simplicity of this integrand representation when extracting 
leading singularities, we explicitly consider in the following one-loop bubble
and triangle Feynman integrals, where for the sake of this analysis, 
we drop overall factors coming from angular integrations.

\paragraph{One-loop bubble integral\\}

The integrand of the one-loop bubble in $D\geq2$ becomes,
\begin{align}
&J^{\left(1\right),D}_2
\left(1,2\right)  
=\frac{1}{2}\,\mathcal{J}_{\left[D_{\parallel}\right]}^{\left(1\right)}
\,\int_{-\infty}^{\infty} da_{1} \int_{0}^{\infty}d\lambda_{11}\,\frac{\lambda_{11}^{(D_{\perp}-2)/2}}{D_{1}D_{2}}\nonumber 
\\
 & =\frac{1}{2}\,\sqrt{-p_{1}^{2}}\,\int_{-\infty}^{\infty} da_{1} \int_{0}^{\infty}d\lambda_{11} \\
&\times \frac{\lambda_{11}^{\left(D_{\perp}-2\right)/2}}{\left(a_{1}^{2}p_{1}^{2}-m_{1}^{2}+\lambda_{11}\right)\left(\left(a_{1}+1\right)^{2}p_{1}^{2}-m_{2}^{2}+\lambda_{11}\right)}\,, \notag
\end{align}
where $p_1^2\ne 0$ and, for this loop topology, $D_\parallel=1$ and $D_\perp=D-1$.

We observe that, by examining the structure of the integrand for different values of $D_\perp = D-1$,
and by relying on a partial fractioning along the lines of the approach presented in~\cite{Henn:2020lye,Wasser:2022kwg}, 
we find, 
\begin{subequations}
\begin{align}
&J_{2}^{\left(1\right),D=2}
\sim \notag \\
&\pm\frac{1}{2\sqrt{\lambda_{\text{K}}\left(p_{1}^{2},m_{1}^{2},m_{2}^{2}\right)}}
\int
d\log\frac{\sqrt{\lambda_{11}}-\sqrt{-D_{2,\parallel}}}{\sqrt{\lambda_{11}}+\sqrt{-D_{2,\parallel}}}  \notag\\
&\land\Bigg(d\log\frac{D_{21}+\sqrt{D_{21}+r_{1}^{+}}+\sqrt{D_{21}+r_{1}^{-}}+\sqrt{r_{1}^{+}}\sqrt{r_{1}^{-}}}{D_{21}+\sqrt{D_{21}+r_{1}^{+}}+\sqrt{D_{21}+r_{1}^{-}}-\sqrt{r_{1}^{+}}\sqrt{r_{1}^{-}}} \notag\\
&-\left(r_{1}^{\pm}\to r_{2}^{\pm}\right)\Bigg)
\,,
\\
&J_{2}^{\left(1\right),D=3}\sim \notag \\
&\frac{1}{\sqrt{-p_{1}^{2}}}\left[\int_{-\infty}^{\infty}\frac{dD_{12}}{D_{12}}\int_{A_{1}}^{\infty}\frac{dD_{1}}{D_{1}}+\int_{-\infty}^{\infty}\frac{dD_{21}}{D_{21}}\int_{A_{2}}^{\infty}\frac{dD_{2}}{D_{2}}\right]\,,
\\
&J_{2}^{\left(1\right),D\geq4}\to\text{ No $d\log$ representation} \label{no_dlog_bubble}\,,
\end{align}
\label{eq:bubbleD}
\end{subequations}
with $A_{1}=\frac{\lambda_{\text{K}}\left(m_{2}^{2}-D_{12},m_{1}^{2},p_{1}^{2}\right)}{4p_{1}^{2}}\,,A_{2}=\frac{\lambda_{\text{K}}\left(m_{1}^{2}-D_{21},m_{2}^{2},p_{1}^{2}\right)}{4p_{1}^{2}}$
and $D_{ij}=D_{i}-D_{j}$.
The prefactor (given only in terms of kinematic invariants)
of the various integrals corresponds to the leading singularity
of a Feynman integral in a particular space-time dimension. The no $d\log$ representation in \eqref{no_dlog_bubble} means that there is a double pole and thus it is not possible to transform it to the logarithmic representation.

Feynman Integrals~\eqref{eq:bubbleD} were obtained with the aid of relations~\eqref{eq:relDlog}, 
summarised in Appendix~\ref{app:landexamples}. 
Also, to make explicit the parallel component of propagators, 
we define, $D_{i,\parallel}\equiv D_{i}-\lambda_{11}$,
that is at most a quadratic polynomial in $a_1$ and absent of $\lambda_{11}$ 
(see Eq.~\eqref{lambda_eq}). 
To simplify the notation, here and in the following,
we use the shorthand notation $D_{ij} = D_{i}-D_{j}$. 
The values of $r^\pm_i$, for $i=1,2$, are, 
\begin{subequations}
\begin{align}
r_{1}^{\pm}=&m_{21}^{2}-p_{1}^{2}\pm2\sqrt{p_{1}^{2}}\sqrt{m_{1}^{2}}\,,
\\
r_{2}^{\pm}=&m_{21}^{2}+p_{1}^{2}\pm2\sqrt{p_{1}^{2}}\sqrt{m_{2}^{2}}\,,
\end{align}
\end{subequations}
in which both of them satisfy the relation $r_i^+ r_i^- = \lambda_\text{K}(p_1^2,m_1^2,m_2^2)$. 

Let us, however, remark that the integrands~\eqref{eq:bubbleD}, 
expressed as products of $d\log${} forms, are often referred to as 
integrands in $d\log$ representation and are by no means unique (see e.g.~\cite{Henn:2021aco}). 
One can find alternative and more compact representations in terms of single $d\log$ integrands,
by accounting for an educated change of variables.
For instance, 
\begin{subequations}
\begin{align}
&J_{2}^{\left(1\right),D=2}
\sim \notag \\
&\pm\frac{1}{2\sqrt{\lambda_{k}\left(p_{1}^{2},m_{1}^{2},m_{2}^{2}\right)}}
\int_{\mathcal{K}}
d\log\left(\frac{D_{1}}{D_{\pm}}\right)\wedge d\log\left(\frac{D_{2}}{D_{\pm}}\right) \label{bubble_dlog_d_2}
\,,
\\
&J_{2}^{\left(1\right),D=3}\sim
\frac{1}{4\sqrt{-p_{1}^{2}}}\int_{\mathcal{C}} d\log D_{1}\land d\log D_{2}
\,,
\label{eq:bubbleD1}
\end{align}
\end{subequations}
where, $D_\pm=(\ell_1-\ell_\pm)^2$, with $\ell_\pm$
either of the two solutions of the maximal cut conditions, $D_1=D_2=0$. The $\pm$ sign in front of \eqref{bubble_dlog_d_2} is related to the solution one has chosen for $D_{\pm}$. The domain of integration $\mathcal{K}$ is a function of external and internal kinematics which we were not able to find explicitly. Its determination remains an open problem. The domain of integration $\mathcal{C}$ is given by the condition $\frac{\det(q_{i} \cdot q_{j})}{\det(p_{1}^{2})}\geq 0$, where $q_{i}$ are expressed in terms of variables $D_{i}$ as $q_{i}^{2} = D_{i} +m_{i}^{2}$, for $i,j=1,2$. $J_{2}^{\left(1\right),D=3}$ can also be written as an iterated integral as
\begin{align}
J_{2}^{\left(1\right),D=3}\sim\int_{-m_{2}^{2}}^{\infty}\frac{dD_{2}}{D_{2}}\int_{D_{2}-m_{1}^{2}+m_{2}^{2}+s-2\sqrt{s\left(D_{2}+m_{2}^{2}\right)}}^{D_{2}-m_{1}^{2}+m_{2}^{2}+s+2\sqrt{s\left(D_{2}+m_{2}^{2}\right)}}\frac{dD_{1}}{D1}\,.
\end{align}

Since the leading singularity turns out to be the same, regardless of the $d\log$ representation
of a given Feynman integral, in the explicit computations,
we draw only our attention to the features of this 
singularity in different space-time dimensions.

\paragraph{One-loop triangle integral\\}

We now continue with the one-loop triangle in $D\geq3$ 
(i.e., $D_\parallel=2$ and $D_\perp=D-2\geq1$),
whose leading singularities become, 
\begin{subequations}
\begin{align}
&J_{3}^{\left(1\right),D=3}\sim \notag \\
&\pm\frac{1}{8\sqrt{-\text{LanS}_{3}^{\left(1\right)}}}
\int_{\textcolor{blue}{\mathcal{K}}} d\log\frac{D_{1}}{D_{\pm}}\land d\log\frac{D_{2}}{D_{\pm}}\land d\log\frac{D_{3}}{D_{\pm}}
\,,
\label{eq:triangleD0}
\\
&J_{3}^{\left(1\right),D=4}\sim \notag \\
&\frac{1}{4\sqrt{\lambda_\text{K}\left(p_{1}^{2},p_{2}^{2},p_{3}^{2}\right)}}
\int_{\textcolor{blue}{\mathcal{C}}} d\log D_{1}\land d\log D_{2}\land d\log D_{3}
\,,
\label{eq:triangleD1}
\\
&J_{3}^{\left(1\right),D\geq5}\to\text{ No $d\log$ representation}\,,\label{no_dlog_triangle}
\end{align}
\label{eq:triangleD}
\end{subequations}
where the ratios $D_i/D_\pm$ 
are meromorphic functions on the variables $a_i$ (with $i=1,2$) and $\lambda_{11}$, 
and, similar to Eq.~\eqref{eq:triangleD0}, $D_\pm=(\ell_1-\ell_\pm)^2$, is obtained
from the on-shell conditions, $D_1=D_2=D_3=0$. The domain of integration $\mathcal{K}$ is a function of external and internal kinematics which we were not able to find explicitly. Its determination remains an open problem. The domain of integration $\mathcal{C}$ is given by by the condition $\frac{\det(q_{i} \cdot q_{j})}{\det(p_{k}\cdot p_{l})}\geq 0$, where $q_{i}$ are expressed in terms of variables $D_{i}$ as $q_{i}^{2} = D_{i} +m_{i}^{2}$, for $i,j=1,2,3$ and $k,l=1,2$.
The leading singularity can exactly be cast in terms of the Landau singularity 
of a one-loop triangle integral~\eqref{eq:land3pt}. The no $d\log$ form in \eqref{no_dlog_triangle}, similarly to the bubble integral \eqref{no_dlog_bubble}, means that there is a double pole.

As anticipated in Sec.~\ref{sec:decomposition}, the decomposition of the
space-time dimension in terms of parallel and perpendicular subspaces
leads to the Baikov representation, which can clearly be appreciated in 
Eqs.~\eqref{eq:bubbleD1} and~\eqref{eq:triangleD1}, respectively, for one-loop 
bubble and triangle in $D=n+1$ space-time dimensions. 

In the following, by making use of algebraic manipulations at the integrand level, 
we demonstrate that leading singularities in $D=n$ and $D=n+1$
space-time dimensions
are, respectively, related to Landau singularities of the first and second type. 

\subsection{Connection between Landau and leading singularities}
\label{eq:lan_lead}

Inspired by the results obtained at one-loop with higher multiplicity
in the previous section and in Appendix~\ref{app:landexamples}, 
it is worth to investigate further the pattern that any one-loop Feynman integral will follow, 
regardless of the space-time dimension. To this end, we exploit the well-defined structure of 
propagators in terms of the loop components $a_{i}$ and $\lambda_{11}$
to elucidate the singular structure of a generic $n$-point scalar one-loop Feynman integral. 

\begin{theorem}
\label{th:lan1}
The leading singularity of an $n$-point one-loop Feynman integral in $D=n+1$ space-time
dimensions is equal to $\pm1/\left(2^{n}\sqrt{-\det\left(p_{i}\cdot p_{j}\right)}\right)$,
with $i,j\leq n-1$.
\end{theorem}
\begin{proof}
This theorem is proven by making use of Lemmas~\ref{lem:wlm1a} and~\ref{lem:wlm1b}. 
In details, we consider the $n$-point Feynman integral~\eqref{eq:int1L} 
in $D=n+1$ space-time dimensions (up to angular integrations),
\begin{align}
J^{\left(1\right),D=n+1}_n
&\sim
\frac{\mathcal{J}_{[n-1]}^{(1)}}{2}
\int\prod_{i=1}^{n-1}da_{i}\,d\lambda_{11}\frac{1}{D_{1}\hdots D_{n}}\,,
\end{align}
whose integrand, because of Lemma~\ref{lem:wlm1a}, can be expressed as, 
\begin{align}
&\int\prod_{i=1}^{n-1}da_{i}\,d\lambda_{11}\frac{1}{D_{1}\hdots D_{n}}
= \notag \\
&\int\sum_{i=1}^{n}\prod_{j=1}^{n-1}da_{j}\,\frac{d\lambda_{11}}{D_{i}}\,
\frac{1}{\prod_{\substack{k=1\\
k\ne i
}
}D_{ki}}\,,
\end{align}
where $D_{ki}=D_{k}-D_{i}$.
The integration over $\lambda_{11}$ clearly manifests a $d\log$ representation, 
whose residue is $+1$. 
The remaining integrations (and residues) are then calculated according to Lemma~\ref{lem:wlm1b}, 
finding, 
\begin{align}
\int\sum_{i=1}^{n}\prod_{j=1}^{n-1}da_{j}\,\frac{d\lambda_{11}}{D_{i}}\,\frac{1}{\prod_{\substack{k=1\\
k\ne i
}
}D_{ki}}
&\sim
\frac{1}{2^{n-1}\det\left(p_{i}\cdot p_{j}\right)}\,,
\label{eq:wth1}
\end{align}
where we identify the variables $x$ of Lemma~\ref{lem:wlm1b}
as, $x_{ij}=2p_i\cdot p_j$.

Therefore, by recalling $\mathcal{J}_{\left[n-1\right]}^{(1)}$ from Eq.~\eqref{eq:jac1L}, 
we find, 
\begin{align}
J^{\left(1\right),D=n+1}_n
&\sim\frac{\pm1}{2^{n}\sqrt{-\det\left(p_{i}\cdot p_{j}\right)}}\,,
\end{align}
completing the proof of this theorem. 
\end{proof}

\begin{theorem}
\label{th:lan2}
The leading singularity of an $n$-point one-loop Feynman integral in $D=n$ space-time dimensions is equal to $ \pm 1/ \left(2^{n}\sqrt{(-1)^{D-1}\rm{LanS}}\right)$ 
\label{th_lan_lead}
\end{theorem}
\begin{proof}
The Feynman integrand in the parallel and perpendicular representation has, in $D=n$, the following form
\begin{equation}
\omega = 
\frac{\mathcal{J}_{\left[D_{\parallel}\right]}^{(1)}}{2}
\frac{1}{\sqrt{\lambda_{11}}}
\frac{d\lambda_{11} \wedge da_{1}\wedge \hdots \wedge da_{n-1}
}{D_{1}\hdots D_{n}}
\,.    
\end{equation}
Let us calculate the Leray residue around  polar sets $D_{1},\hdots , D_{n}$, 
where each $D_{i}$ is given by~\eqref{lambda_eq}, 
\begin{equation}
res^{n}[\omega] = \frac{\mathcal{J}_{\left[D_{\parallel}\right]}^{(1)}}{2}
\frac{1}{\sqrt{\lambda_{11}}}
\frac{d\lambda_{11} \wedge da_{1}\wedge \hdots \wedge da_{n-1}}{dD_{1}\wedge \hdots \wedge dD_{n}}
\,.     
\end{equation}
Each differential $dD_{i}$ has the following form
\begin{align}
&dD_{i} =  \notag \\
&= 2 \sum_{j=1}^{n-1}\sum_{k=1}^{n-1}p_{j} \cdot p_{k} a_{j} da_{k} + 2\sum_{j=1}^{i-1}\sum_{k=1}^{n-1}p_{j}\cdot p_{k} da_{k} + d\lambda_{11} 
\notag
\\
&= \sum_{k=1}^{n-1} \left( 2 \sum_{j=1}^{n-1}p_{j} \cdot p_{k} a_{j} + 2\sum_{j=1}^{i-1}p_{j}\cdot p_{k} \right) da_{k} + d\lambda_{11}\,.
\end{align}
Thus,
\begin{widetext}
\begin{equation}
\begin{split}
&res^{n}[\omega] = \frac{\mathcal{J}_{\left[D_{\parallel}\right]}^{(1)}}{2}
\frac{1}{\sqrt{\lambda_{11}}} 
\frac{d\lambda_{11} \wedge da_{1}\wedge \hdots \wedge da_{n-1}}{\bigg[d\lambda_{11} + \sum\limits_{k=1}^{n-1} \Big( 2 \sum\limits_{j=1}^{n-1}p_{j} \cdot p_{k} a_{j} \Big) da_{k} \bigg] \wedge \hdots \wedge \bigg[d\lambda_{11} + \sum\limits_{k=1}^{n-1} \Big( 2 \sum\limits_{j=1}^{n-1}p_{j} \cdot p_{k} a_{j} + 2\sum\limits_{j=1}^{n-1}p_{j}\cdot p_{k} \Big) da_{k} \bigg]}
\,. 
\end{split}
\end{equation}
\end{widetext}
By Lemma~\ref{wedge_det} this is equal to,
\begin{equation}
\begin{split}
res^{n}[\omega] &= \frac{\mathcal{J}_{\left[D_{\parallel}\right]}^{(1)}}{2}
\frac{1}{\sqrt{\lambda_{11}}} 
\frac{d\lambda_{11} \wedge da_{1}\wedge \hdots \wedge da_{n-1}}{\det(B)d\lambda_{11} \wedge da_{1}\wedge \hdots \wedge da_{n-1}} \\
&=\frac{\mathcal{J}_{\left[D_{\parallel}\right]}^{(1)}}{2}
\frac{1}{\sqrt{\lambda_{11}}} 
\frac{1}{\det(B)}
\,,
\end{split}
\end{equation}
where,
\begin{widetext}
\begin{equation}
\det(B) = 
\begin{vmatrix}
1 & 2 \sum\limits_{j=1}^{n-1}p_{j} \cdot p_{1} a_{j} & \hdots & 2 \sum\limits_{j=1}^{n-1}p_{j} \cdot p_{n-1} a_{j} \\
1 & 2 \sum\limits_{j=1}^{n-1}p_{j} \cdot p_{1} a_{j} + 2 p_{1}\cdot p_{1} & \hdots & 2 \sum\limits_{j=1}^{n-1}p_{j} \cdot p_{n-1} a_{j} + 2 p_{1}\cdot p_{n-1} \\
\vdots & \dots & \vdots & \vdots \\ 
1 & 2 \sum\limits_{j=1}^{n-1}p_{j} \cdot p_{1} a_{j} + 2 \sum\limits_{j=1}^{n-2}p_{j}\cdot p_{1} & \hdots & 2 \sum\limits_{j=1}^{n-1}p_{j} \cdot p_{n-1} a_{j} + 2 \sum\limits_{j=1}^{n-2}p_{j}\cdot p_{n-1} \\
1 & 2 \sum\limits_{j=1}^{n-1}p_{j} \cdot p_{1} a_{j} + 2 \sum\limits_{j=1}^{n-1}p_{j}\cdot p_{1} & \hdots & 2 \sum\limits_{j=1}^{n-1}p_{j} \cdot p_{n-1} a_{j} + 2 \sum\limits_{j=1}^{n-1}p_{j}\cdot p_{n-1}
\end{vmatrix}
\,.
\end{equation}
\end{widetext}
Let us calculate $b_{(i+1)k}-b_{ik}$ with $i,k=1,\hdots , n-1$,
\begin{align}
&\left( 2 \sum_{j=1}^{n-1}p_{j} \cdot p_{k} a_{j} + 2 \sum_{j=1}^{i}p_{j}\cdot p_{k} \right) \notag \\
&- \left( 2 \sum_{j=1}^{n-1}p_{j} \cdot p_{k} a_{j} + 2 \sum_{j=1}^{i-1}p_{j}\cdot p_{k} \right) = 2 p_{i} \cdot p_{k}
\,.
\end{align}
Let us now successively subtract $i$-th row from the $(i+1)$-th row starting from the last row. 
After the first subtraction, we will get,
\begin{widetext}
\begin{equation}
\det(B) = 
\begin{vmatrix}
1 & 2 \sum\limits_{j=1}^{n-1}p_{j} \cdot p_{1} a_{j} & \hdots & 2 \sum\limits_{j=1}^{n-1}p_{j} \cdot p_{n-1} a_{j} \\
1 & 2 \sum\limits_{j=1}^{n-1}p_{j} \cdot p_{1} a_{j} + 2 p_{1}\cdot p_{1} & \hdots & 2 \sum\limits_{j=1}^{n-1}p_{j} \cdot p_{n-1} a_{j} + 2 p_{1}\cdot p_{n-1} \\
\vdots & \dots & \vdots & \vdots \\ 
1 & 2 \sum\limits_{j=1}^{n-1}p_{j} \cdot p_{1} a_{j} + 2 \sum\limits_{j=1}^{n-2}p_{j}\cdot p_{1} & \hdots & 2 \sum\limits_{j=1}^{n-1}p_{j} \cdot p_{n-1} a_{j} + 2 \sum\limits_{j=1}^{n-2}p_{j}\cdot p_{n-1} \\
0 & 2 p_{n-1}\cdot p_{1} & \hdots & 2 p_{n-1}\cdot p_{n-1}
\end{vmatrix}
\,,
\end{equation}
\end{widetext} 
and by repeating this until we reach the first row we will end up with, \\
\begin{align}
\det(B) = 
&\begin{vmatrix}
1 & 2 \sum\limits_{j=1}^{n-1}p_{j} \cdot p_{1} a_{j} & \hdots & 2 \sum\limits_{j=1}^{n-1}p_{j} \cdot p_{n-1} a_{j} \\
0 & 2 p_{1}\cdot p_{1} & \hdots & 2 p_{1}\cdot p_{n-1} \\
\vdots & \dots & \vdots & \vdots \\ 
0 & 2 p_{n-2}\cdot p_{1} & \hdots & 2 p_{n-2}\cdot p_{n-1} \\
0 & 2 p_{n-1}\cdot p_{1} & \hdots & 2 p_{n-1}\cdot p_{n-1}
\end{vmatrix} \notag \\
=&\begin{vmatrix}
1 & 2 \sum\limits_{j=1}^{n-1}p_{j} \cdot p_{1} a_{j} & \hdots & 2 \sum\limits_{j=1}^{n-1}p_{j} \cdot p_{n-1} a_{j} \\
0 & 2 p_{1}\cdot p_{1} & \hdots & 2 p_{1}\cdot p_{n-1} \\
\vdots & \dots & \vdots & \vdots \\ 
0 & 2 p_{1}\cdot p_{n-2} & \hdots & 2 p_{n-2}\cdot p_{n-1} \\
0 & 2 p_{1}\cdot p_{n-1} & \hdots & 2 p_{n-1}\cdot p_{n-1}
\end{vmatrix}
\,.
\end{align}
Let us expand this determinant with respect to the first column,
\begin{align}
\det(B) &= (-1)^{2} 2^{n-1}
\begin{vmatrix}
p_{1}\cdot p_{1} & \hdots & p_{1}\cdot p_{n-1} \\
\vdots & \dots & \vdots \\ 
p_{1}\cdot p_{n-2} & \hdots & p_{n-2}\cdot p_{n-1} \\
p_{1}\cdot p_{n-1} & \hdots & p_{n-1}\cdot p_{n-1}
\end{vmatrix}
\notag
\\ &=
2^{n-1} \det(p_{i} \cdot p_{j})\,.
\end{align}
Therefore, the residue of $\omega$ is equal to,
\begin{equation}
res^{n}[\omega] =
\frac{\mathcal{J}_{\left[D_{\parallel}\right]}^{(1)}}{2}
\frac{1}{\sqrt{\lambda_{11}}} 
\frac{1}{2^{n-1} \det(p_{i} \cdot p_{j})}
\,.
\end{equation}
However, since we have calculated residues around polar sets $D_{1}= \hdots =D_{n}=0$ 
and we showed in Sec.~\ref{sec:landau1L} that under these conditions,
\begin{equation}
\lambda_{11} = \frac{\rm{LanS}}{\det(p_{i}\cdot p_{j})}\,,
\end{equation}
we can insert this into our residue formula to obtain,
\begin{equation}
res^{n}[\omega] =
\frac{\mathcal{J}_{\left[D_{\parallel}\right]}^{(1)}}{2^{n}}
\frac{1}{\sqrt{\rm{LanS}}} 
\frac{\sqrt{\det(p_{i}\cdot p_{j})}}{\det(p_{i} \cdot p_{j})}
\,.      
\end{equation}
Taking into account that $\mathcal{J}_{\left[D_{\parallel}\right]}^{(1)}=\pm\sqrt{(-1)^{D-1}\det(p_{i}\cdot p_{j})}$, we get,
\begin{equation}
\begin{split}
&res^{n}[\omega] = \\
&=\frac{\pm\sqrt{(-1)^{D-1}\det(p_{i}\cdot p_{j})}}{2^{n}}
\frac{1}{\sqrt{\rm{LanS}}} 
\frac{\sqrt{\det(p_{i}\cdot p_{j})}}{\det(p_{i} \cdot p_{j})}  \\
&= \frac{1}{2^{n}}
\frac{\pm 1}{\sqrt{(-1)^{D-1}\rm{LanS}}}
\,.      
\end{split}
\end{equation}
\end{proof}

We notice from the examples listed in this section (and in Appendix~\ref{app:landexamples}) that one-loop scalar Feynman integrals in $D=n$ and $D=n+1$ space-time dimensions manifest a $d\log$ representation, where the arguments of $d\log$'s have very simple structure.
In effect, combining this observation with the above results for the one-loop scalar integrals we make the following conjecture.\footnote{An evidence of a similar result in a dual formulation of planar and non-planar $\mathcal{N}=4\,\text{sYM}$ was pointed out in~\cite{Arkani-Hamed:2014via,Bern:2014kca}.}
\begin{con}
An n-point Feynman integral can be written in one of the following forms depending on the space-time dimension $D$,
\begin{subequations}
\begin{align}
&J^{\left(1\right),D=n}_n 
\sim \notag \\
&\pm\frac{1}{2^{n}\sqrt{-{\rm LanS}_n^{(1)}}}
\int_{\mathcal{K}} d\log\frac{D_{1}}{D_{\pm}}\wedge\hdots\wedge d\log\frac{D_{n}}{D_{\pm}}
\,, \label{con_d_n} \\[0.3cm]
&J^{\left(1\right),D=n+1}_n 
\sim \notag \\
&\frac{1}{2^{n}\sqrt{-\det\left(p_{i}\cdot p_{j}\right)}}
\int_{\mathcal{C}} d\log D_{1}\wedge\hdots\wedge d\log D_{n}
\,, \notag \\
&\text{ with $i,j\leq n-1$}\,,
\\[0.3cm]
&J^{\left(1\right),D\geq n+2}_n 
\to\text{ No $d\log$ representation}\,, \label{con_no_dlog}
&&
\end{align}
\label{eq:myleadex}
\end{subequations}
where $\det\left(p_{i}\cdot p_{j}\right)$ can be identified with the Landau singularity of the second type, and $D_\pm$ are calculated by imposing the on-shell conditions $D_1=D_2=\hdots=D_n=0$. The $\pm$ sign in front of \eqref{con_d_n} is related to the solution one has chosen for $D_{\pm}$. 
The domain of integration $\mathcal{K}$ is a function of external and internal kinematics which we were not able to find explicitly. Its determination remains an open problem. The domain of integration $\mathcal{C}$ is given by by the condition $\frac{\det(q_{i} \cdot q_{j})}{\det(p_{k}\cdot p_{l})}\geq 0$, where $q_{i}$ are expressed in terms of variables $D_{i}$ as $q_{i}^{2} = D_{i} +m_{i}^{2}$, for $i,j=1,\hdots,n$ and $k,l=1,\hdots, n-1$. In other words $\mathcal{C}$ is given by the condition $\lambda_{11} \geq 0$, similarly to \eqref{lambda_landau} without imposing on-shell conditions and expressing it in terms of variables $D_{i}$. The no $d\log$ form means that there is a double pole and it cannot be transformed to logarithmic representation.
\end{con}

\subsection{Differential equations of one-loop Feynman integrals}
\label{sec:deq1L}

The connection between leading and Landau singularities in one-loop Feynman 
integrals motivates us to continue our analysis with the calculation of Feynman 
integrals through differential equations methods. 
As revealed in Ref.~\cite{Henn:2013pwa}, the difficulty in this analytic calculation 
can be largely reduced once the differential equation is in the canonical form. 
In the following, we show that from the knowledge of the minimal set of integrals (often referred to as
master integrals) in Laporta basis~\cite{Laporta:2000dsw}, one can generate differential equations 
in the canonical form by appropriately accounting for leading singularities of these integrals 
in specific space-time dimensions.

Since scalar integrals that are chosen as master integrals might display leading singularities in specific 
space-time dimensions ($D=n$ or $D=n+1$), 
one can profit from dimension-shift relations in Feynman integrals to pass from $D\to D\pm 2$~\cite{Tarasov:1996br}.
We describe this observation with the calculation of the differential equation in canonical form 
of the set of master integrals present in the four-point integrals with massive external 
and internal momenta. 
We choose to reduce the number of master integrals with off-shell momenta $p_i=m^2$ 
and internal massive propagators with mass $m^2$. 

\begin{figure}[t]
\centering
\subfigure[]{\label{fig:4ptmis-a}\includegraphics[scale=0.7]{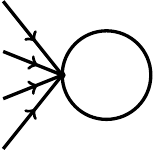}}
\qquad
\subfigure[]{\label{fig:4ptmis-b}\includegraphics[scale=0.7]{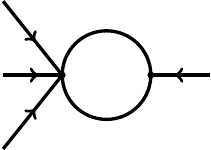}}
\qquad
\subfigure[]{\label{fig:4ptmis-c}\includegraphics[scale=0.7]{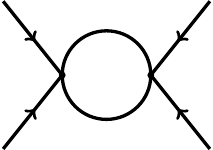}}
\qquad
\subfigure[]{\label{fig:4ptmis-d}\includegraphics[scale=0.7]{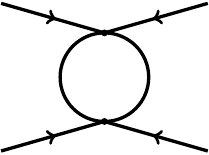}}
\qquad\\
\subfigure[]{\label{fig:4ptmis-e}\includegraphics[scale=0.7]{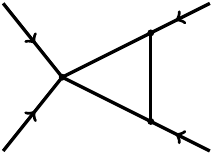}}
\qquad
\subfigure[]{\label{fig:4ptmis-f}\includegraphics[scale=0.7]{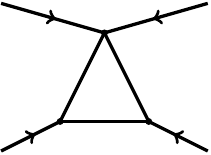}}
\qquad
\subfigure[]{\label{fig:4ptmis-g}\includegraphics[scale=0.7]{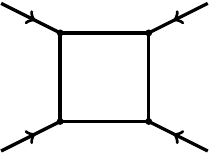}}
\qquad
\caption{Four-point one-loop master integrals.}
\label{fig:4ptmis}
\end{figure}

By following the Laporta method for the generation of integration-by-parts identities, 
we find that any four-point integral can be expressed in terms of seven independent master integrals, 
\begin{align}
\Big\{ 
J_{1}^{(1),D}(1)\,,
J_{2}^{(1),D}(1,4)\,,
J_{2}^{(1),D}(1,3)\,,
J_{2}^{(1),D}(2,4)\,,
\notag\\
J_{3}^{(1),D}(1,2,4)\,,
J_{3}^{(1),D}(1,2,3)\,,
J_{4}^{(1),D}(1,2,3,4)
\Big\}\,,
\label{eq:4pt_mi_lap}
\end{align}
according to convention~\eqref{eq:intL} and, respectively, displayed in Fig.~\ref{fig:4ptmis}.

Let us now turn our attention to the evaluation of the integrals in $D=4-2\epsilon$.
From the analysis of Sec.~\ref{eq:lan_lead}, we know that not all master integrals~\eqref{eq:4pt_mi_lap}
manifest a leading singularity in $D=4$. For tadpole and bubble integrals, however, 
we can make use of dimensional recurrence relations to consider these
integrals in $D=2-2\epsilon$,
\begin{align}
&J_{1}^{\left(1\right),D-2}\left(1\right)  =\frac{\left(D-2\right)}{2m^{2}}J_{1}^{\left(1\right),D}\left(1\right)\,,\nonumber \\
&J_{2}^{\left(1\right),D-2}\left(i,j\right)  =\frac{\left(D-2\right)}{3m^{4}}J_{1}^{\left(1\right),D}\left(1\right)+\frac{2\left(D-3\right)}{3m^{2}}J_{2}^{\left(1\right),D}\left(i,j\right)\,.
\label{eq:relD1}
\end{align}
Thus, by normalising leading singularities to $1$ (in the sense that irrelevant overall numerical factors
independent of kinematic scales are dropped), 
we choose,\footnote{We generated integration-by-parts identities with
the aid of the publicly available software \textsc{Fire6}~\cite{Smirnov:2019qkx},
whereas for differential equations w.r.t. $s$, $t$, and $m^2$ 
we employed the \textsc{LiteRed} framework~\cite{Lee:2012cn}.} 
\begin{align}
g_{1} & =\epsilon\,J_{1}^{(1),D=2-2\epsilon}\left(1\right)\,,\nonumber \\
g_{2} & =\epsilon\,m^{2}\,J_{2}^{(1),D=2-2\epsilon}\left(1,4\right)\,,\nonumber \\
g_{3} & =\epsilon\,\sqrt{-s\left(4m^{2}-s\right)}\,J_{2}^{(1),D=2-2\epsilon}\left(1,3\right)\,,\nonumber \\
g_{4} & =\epsilon\,\sqrt{-t\left(4m^{2}-t\right)}\,J_{2}^{(1),D=2-2\epsilon}\left(2,4\right)\,,\nonumber \\
g_{5} & =\epsilon^{2}\,\sqrt{-s\left(4m^{2}-s\right)}\,J_{3}^{(1),D=4-2\epsilon}\left(1,2,4\right)\,,\nonumber \\
g_{6} & =\epsilon^{2}\,\sqrt{-t\left(4m^{2}-t\right)}\,J_{3}^{(1),D=4-2\epsilon}\left(1,2,3\right)\,,\nonumber \\
g_{7} & =\epsilon^{2}\,\sqrt{st\left(12m^{4}-4m^{2}\left(s+t\right)+st\right)}\,J_{4}^{(1),D=4-2\epsilon}\left(1,2,3,4\right)\,.
\end{align}
that automatically satisfies the differential equation in the canonical form, 
\begin{align}
\partial_{\eta}\,\vec{g}&=\epsilon\, A_{\eta}\vec{g}\,,
\label{eq:deqcan}
\end{align}
with $\eta\in\{s,t,m^2\}$ and $\vec{g}=\{g_1,\hdots,g_7\}$. 
For instance, we find for $g_7$ the differential equation w.r.t. $s$,
\begin{align}
\partial_{s}g_{7}= & \frac{\epsilon}{m^{2}}\Bigg[-\frac{6R_{22}}{\left(R_{1}^{2}+1\right)R_{3}R_{11}}g_{2}+\frac{R_{22}}{R_{1}\left(R_{1}^{2}+1\right)R_{3}}g_{3}\nonumber \\
 & +\frac{R_{2}}{R_{3}R_{11}}g_{4}-\frac{2\left(R_{2}^{2}R_{1}^{2}+R_{1}^{2}-2\right)R_{22}}{R_{1}\left(R_{1}^{2}+1\right)\left(R_{1}^{2}+R_{2}^{2}+4\right)R_{3}}g_{5}\nonumber \\
 & -\frac{2R_{2}\left(R_{2}^{2}+2\right)}{\left(R_{1}^{2}+R_{2}^{2}+4\right)R_{3}R_{11}}g_{6}-\frac{\left(R_{2}^{2}+2\right){}^{2}}{\left(R_{1}^{2}+R_{2}^{2}+4\right)R_{3}^{2}}g_{7}\Bigg]\,,
\end{align}
with, 
\begin{align}
 & R_{11}^{2}=s/m^{2}\,, &  & R_{1}^{2}=R_{11}^{2}-4\,,\nonumber \\
 & R_{22}^{2}=t/m^{2}\,, &  & R_{2}^{2}=R_{22}^{2}-4\,,\nonumber \\
 & R_{3}^{2}=R_{1}^{2}R_{2}^{2}-4\,.
\end{align}
Similar structures are found for the other derivatives. 

\paragraph{Differential equation in $D=5-2\epsilon$\\}

In view that our method does not introduce any limitation w.r.t. the choice of the space-time dimension,
we consider as an illustrative example the construction of differential equation in canonical 
form of the four-point integral family in $D=5-2\epsilon$. 

We begin by selecting the space-time dimension where master integrals~\eqref{eq:4pt_mi_lap} 
are known to manifest leading singularities. Besides relations~\eqref{eq:relD1}, 
we consider the following relation for the one-loop scalar triangle, 
\begin{align}
\label{eq:relD2}
&J_{3}^{\left(1\right),D-2}\left(1,2,k\right)=  \notag \\
&\frac{\left(D-2\right)\left(14m^{2}-5Q_{k}\right)}{6m^{4}\left(3m^{2}-Q_{k}\right)\left(4m^{2}-Q_{k}\right)}J_{1}^{\left(1\right),D}\left(1\right)\nonumber\\
 & +\frac{\left(D-3\right)\left(2m^{2}-Q_{k}\right)}{m^{2}\left(3m^{2}-Q_{k}\right)\left(4m^{2}-Q_{k}\right)}J_{2}^{\left(1\right),D}\left(2,k\right) \nonumber \\
 & +\frac{2\left(D-3\right)}{3m^{2}\left(3m^{2}-Q_{k}\right)}J_{2}^{\left(1\right),D}\left(1,k\right) \nonumber\\
 &+\frac{\left(D-4\right)\left(4m^{2}-Q_{k}\right)}{2m^{2}\left(3m^{2}-Q_{k}\right)}J_{2}^{\left(1\right),D}\left(1,2,k\right)\,, 
\end{align}
for $k=3,4$, and $Q_3=s$ and $Q_4=t$.

Thus, with relations~\eqref{eq:relD1} and~\eqref{eq:relD2}, we choose, 
\begin{align}
g_{1} & =\epsilon\,\sqrt{m^{2}}\,J_{1}^{(1),D=1-2\epsilon}\left(1\right)\,,\nonumber \\
g_{2} & =\epsilon^{2}\,\sqrt{m^{2}}\,J_{2}^{(1),D=3-2\epsilon}\left(1,4\right)\,,\nonumber \\
g_{3} & =\epsilon^{2}\,\sqrt{-s}\,J_{2}^{(1),D=3-2\epsilon}\left(1,3\right)\,,\nonumber \\
g_{4} & =\epsilon^{2}\,\sqrt{-t}\,J_{2}^{(1),D=3-2\epsilon}\left(2,4\right)\,,\nonumber \\
g_{5} & =\epsilon^{2}\,\sqrt{-sm^{2}\left(3m^{2}-s\right)}\,J_{3}^{(1),D=3-2\epsilon}\left(1,2,4\right)\,,\nonumber \\
g_{6} & =\epsilon^{2}\,\sqrt{-tm^{2}\left(3m^{2}-t\right)}\,J_{3}^{(1),D=3-2\epsilon}\left(1,2,3\right)\,,\nonumber \\
g_{7} & =\epsilon^{3}\,\sqrt{-st\left(4m^{2}-s-t\right)}\,J_{4}^{(1),D=5-2\epsilon}\left(1,2,3,4\right)\,.
\end{align}
finding differential equations w.r.t. $s,t$, and $m^2$ in canonical form~\eqref{eq:deqcan}.
 
\par\bigskip In these illustrative examples, we showed that from the knowledge of master integrals 
in terms of one-loop scalar integrals, the construction of differential equations
in the canonical form is automatic from the clear knowledge of leading singularities in specific 
space-time dimensions. 

The results of this section show that the construction of the canonical form of the differential equation for one-loop Feynman integrals is completely determined by the Landau analysis of master integrals.
Allowing, in this way, to focus more on the interesting properties of Feynman integrals 
given by their singularities.

\section{Preliminaries on multi-loop Feynman integrals}
\label{sec:twoloop} 
To this end, in this section, we explicitly showcase computations of particular two-loop 
Feynman integrals 
(with internal massless propagators and external massive momenta) 
whose leading singularities can be extracted along the lines 
of the approach carried out for one-loop Feynman integrals. 
Additionally, as a by-product of the explicit calculations of two-loop planar and non-planar
triangles, we make use of the loop-by-loop approach and multi-dimensional residues to provide an alternative proof of
analytic expressions
for the four-dimensional leading singularities of $L$-loop ladder triangle and box integrals, 
whose expressions are known for a long time~\cite{Usyukina:1993ch}.

\subsection{Application of the Leray residue at two loops}
\label{sec:leray2L}

\begin{figure}[t]
\centering
\subfigure[]{\label{fig:tri2L-a}\includegraphics[scale=0.75]{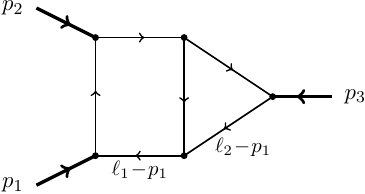}}
\quad
\subfigure[]{\label{fig:tri2L-b}\includegraphics[scale=0.75]{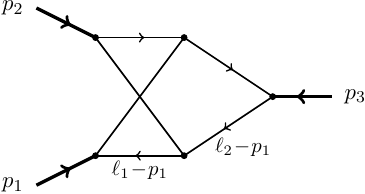}}
\quad
\subfigure[]{\label{fig:box2L-c}\includegraphics[scale=0.75]{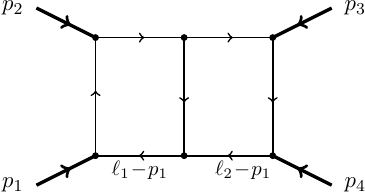}}
\caption{Two-loop planar and non-planar topologies with massless propagators 
and massive external momenta, $p_i^2\ne 0$.}
\label{fig:tri2L}
\end{figure}

\paragraph{Two-loop planar triangle\\}
 
Let us start considering the two-loop topology displayed in Fig.~\ref{fig:tri2L-a},
whose Feynman integral with internal massless propagators can be described as, 
\begin{align}
&J_{\text{P;$3$-pt}}^{(2),D=4} = 
\int\frac{d^{4}\ell_{1}}{\imath\pi^{2}}\frac{d^{4}\ell_{2}}{\imath\pi^{2}}  \\
&\times \frac{1}{\ell_{1}^{2}\left(\ell_{1}-\ell_{2}\right)^{2}\left(\ell_{1}-p_{1}\right)^{2}\left(\ell_{2}-p_{1}\right)^{2}\left(\ell_{1}+p_{2}\right)^{2}\left(\ell_{2}+p_{2}\right)^{2}}\,,
\label{eq:2L_planar}
\end{align}
with the same convention for external kinematics ($p_i^2\ne0$) utilised in Secs.~\ref{sec:landau}
and~\ref{sec:leading}. 

Then, by taking into account the decomposition of space-time dimension into 
parallel ($D_\parallel=2$) and perpendicular ($D_\perp=2$) directions, along 
the lines of Sec.~\ref{sec:deco_multiloop}, we parametrize the loop momenta as, 
\begin{align}
\ell_{i}^{\alpha}&=a_{i1}\,p_{1}^{\alpha}+a_{i2}\,p_{2}^{\alpha}+\lambda_{i}^{\alpha}\,,
\label{eq:tri2L_loopmom}
\end{align}
finding for the integrand, 
\begin{align}
J_{\text{P;$3$-pt}}^{(2),D=4}\sim
&\pm\frac{1}{16}\lambda_{\text{K}}\left(p_{1}^{2},p_{2}^{2},p_{3}^{2}\right) 
\notag \\
& \int\frac{da_{11}\,da_{12}\,da_{21}\,da_{22}\,d\lambda_{11}\,d\lambda_{22}\,d\lambda_{12}}{\sqrt{\lambda_{11}\lambda_{22}-\lambda_{12}^{2}}\,\ell_{1}^{2}\left(\ell_{1}-\ell_{2}\right)^{2}\cdots\left(\ell_{2}+p_{2}\right)^{2}}\,,
\end{align}
where, as in the examples at one-loop, we omit overall terms that arise from angular integrations. 
Notice that propagators are now expressed in terms of the loop-momentum 
components $a_{ij}$ and $\lambda_{ij}$, e.g.,
\begin{align}
\left(\ell_{1}-\ell_{2}\right)^{2}&=
p_{1}^{2}\left(a_{11}-a_{21}\right)\left(a_{11}-a_{12}-a_{21}+a_{22}\right) \notag \\
&+p_{2}^{2}\left(a_{12}-a_{22}\right)\left(-a_{11}+a_{12}+a_{21}-a_{22}\right)
\notag\\
&+p_{3}^{2}\left(a_{11}-a_{21}\right)\left(a_{12}-a_{22}\right)
+\lambda_{11}-2\lambda_{12}+\lambda_{22}\,.
\end{align}
Then, the residue takes the form,\footnote{The calculation of residues through Leray's theory
has been implemented in \Mathematica{} through its built-in function \texttt{TensorWedge}.} 
\begin{align}
&J_{\text{P;$3$-pt}}^{(2),D=4}
\sim \pm\frac{1}{16\lambda_{\text{K}}\left(p_{1}^{2},p_{2}^{2},p_{3}^{2}\right)} \notag \\
&\times \frac{\sqrt{\lambda_{11}\lambda_{22}-\lambda_{12}^{2}}}{\left(\lambda_{12}\left(a_{11}-a_{12}-1\right)+\lambda_{11}\left(-a_{21}+a_{22}+1\right)\right)}\,,
\end{align}
where we are left to find the values of the integration variables by accounting 
from the seven conditions,
\begin{align}
\sqrt{\lambda_{11}\lambda_{22}-\lambda_{12}^{2}}=\ell_{1}^{2}=\left(\ell_{1}-\ell_{2}\right)^{2}=\cdots=\left(\ell_{2}+p_{2}\right)^{2}=0\,,
\end{align} 
finding in this way, 
\begin{align}
J_{\text{P;$3$-pt}}^{(2),D=4}&
\sim\pm\frac{1}{16p_{3}^{2}\sqrt{\lambda_{\text{K}}\left(p_{1}^{2},p_{2}^{2},p_{3}^{2}\right)}}\,,
\end{align}
that is in complete agreement with the leading singularity delivered by the automated package \Dlog{}.

\paragraph{Two-loop non-planar triangle\\}

Let us now draw our attention to the non-planar triangle of Fig.~\ref{fig:tri2L-b}, 
whose Feynman integral with internal massless propagators can be characterised as,   
\begin{align}
&J_{\text{NP;\ensuremath{3}-pt}}^{(2),D=4}=\int\frac{d^{4}\ell_{1}}{\imath\pi^{2}}\frac{d^{4}\ell_{2}}{\imath\pi^{2}} \times \notag \\
& \frac{1}{\ell_{1}^{2}\left(\ell_{1}-p_{1}\right)^{2}\left(\ell_{2}-p_{1}\right)^{2}\left(\ell_{2}+p_{2}\right)^{2}\left(\ell_{1}-\ell_{2}\right)^{2}\left(\ell_{1}-\ell_{2}-p_{2}\right)^{2}}\,.
\end{align}

By parametrizing the loop momenta according to~\eqref{eq:tri2L_loopmom}
and following the same procedure carried out in the calculation of the two-loop planar triangle, 
we find the residue,
\begin{align}
&J_{\text{NP;\ensuremath{3}-pt}}^{(2),D=4}\sim\pm\frac{1}{16\lambda_{\text{K}}\left(p_{1}^{2},p_{2}^{2},p_{3}^{2}\right)} \notag \\
&\frac{\sqrt{\lambda_{11}\lambda_{22}-\lambda_{12}^{2}}}{\lambda_{11}\left(a_{21}-a_{22}-1\right)-\lambda_{12}\left(a_{11}-a_{12}-a_{22}-1\right)-a_{12}\lambda_{22}}\,,
\end{align}
in which, the loop-momentum components $a_{ij}$ and $\lambda_{ij}$ are fixed from 
the conditions, 
\begin{align}
&\sqrt{\lambda_{11}\lambda_{22}-\lambda_{12}^{2}}=\ell_{1}^{2}=\left(\ell_{1}-p_{1}\right)^{2} \notag \\
&=\cdots=\left(\ell_{1}-\ell_{2}-p_{2}\right)^{2}=0\,,
\end{align}
amounting to, 
\begin{align}
J_{\text{NP;\ensuremath{3}-pt}}^{(2),D=4}&\sim\pm\frac{1}{16\lambda_{\text{K}}\left(p_{1}^{2},p_{2}^{2},p_{3}^{2}\right)}\,,
\end{align}
finding complete agreement with the result delivered by \Dlog{}.

\subsection{Loop-by-loop approach}
\label{sec:loopbyloop}

\begin{figure}[t]
\centering
\subfigure[]{\label{fig:ladders-a}\includegraphics[scale=0.75]{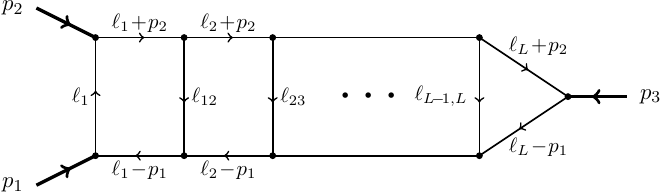}}
\qquad
\subfigure[]{\label{fig:ladders-b}\includegraphics[scale=0.75]{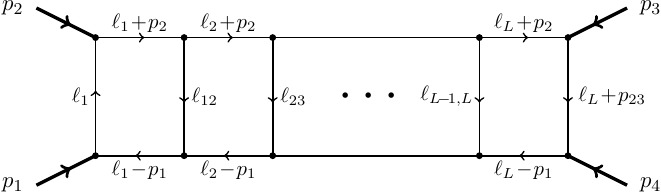}}
\caption{$L$-loop planar ladder triangle and box topologies with massless propagators 
and massive external momenta, $p_i^2\ne0$. 
We define, for convenience, $\ell_{ij}=\ell_i-\ell_j$ and $p_{23}=p_2+p_3$. }
\label{fig:ladders}
\end{figure}

In view of the 
connection between leading and Landau singularities considered in Sec.~\ref{sec:leading},
one can exploit this connection by following a loop-by-loop approach based on 
the former results.
To this end, we discuss this approach in the two-loop triangles previously considered
and, as a by-product of these results, 
give the explicit expression of the $L$-loop ladder planar triangle and box scalar integrals. 

\paragraph{Two-loop planar triangle and box integrals\\}

Let us then begin with $J_{\text{P;$3$-pt}}^{(2),D=4}$ of Eq.~\eqref{eq:2L_planar},
which, by separately performing loop integrations, according to Fig.~\ref{fig:tri2L-a}, 
left-right (first in $\ell_1$ and then in $\ell_2$)
and 
right-left (first in $\ell_2$ and then in $\ell_1$), respectively, becomes, 
\begin{widetext}
\begin{subequations}
\begin{align}
J_{\text{P;\ensuremath{3}-pt}}^{(2),D=4}\Big|_{\text{L-R}}= & \int\frac{d^{4}\ell_{2}}{\imath\pi^{2}}\frac{1}{\left(\ell_{2}-p_{1}\right)^{2}\left(\ell_{2}+p_{2}\right)^{2}}\int\frac{d^{4}\ell_{1}}{\imath\pi^{2}}\frac{1}{\ell_{1}^{2}\left(\ell_{1}-\ell_{2}\right)^{2}\left(\ell_{1}-p_{1}\right)^{2}\left(\ell_{1}+p_{2}\right)^{2}}\nonumber \\
\sim & \int\frac{d^{4}\ell_{2}}{\imath\pi^{2}}\frac{1}{\ell_{2}^{2}\left(\ell_{2}-p_{1}\right)^{2}\left(\ell_{2}+p_{2}\right)^{2}p_{3}^{2}\sqrt{\lambda_{\text{K}}\left(1,\frac{p_{1}^{2}\left(\ell_{2}+p_{2}\right)^{2}}{p_{3}^{2}\ell_{2}^{2}},\frac{p_{2}^{2}\left(\ell_{2}-p_{1}\right)^{2}}{p_{3}^{2}\ell_{2}^{2}}\right)}}\,,\\
J_{\text{P;\ensuremath{3}-pt}}^{(2),D=4}\Big|_{\text{R-L}}= & \int\frac{d^{4}\ell_{1}}{\imath\pi^{2}}\frac{1}{\ell_{1}^{2}\left(\ell_{1}-p_{1}\right)^{2}\left(\ell_{1}+p_{2}\right)^{2}}\int\frac{d^{4}\ell_{2}}{\imath\pi^{2}}\frac{1}{\left(\ell_{1}-\ell_{2}\right)^{2}\left(\ell_{2}-p_{1}\right)^{2}\left(\ell_{2}+p_{2}\right)^{2}}\nonumber \\
\sim & \int\frac{d^{4}\ell_{1}}{\imath\pi^{2}}\frac{1}{\ell_{1}^{2}\left(\ell_{1}-p_{1}\right)^{2}\left(\ell_{1}+p_{2}\right)^{2}p_{3}^{2}\sqrt{\lambda_{\text{K}}\left(1,\frac{\left(\ell_{1}-p_{1}\right)^{2}}{p_{3}^{2}},\frac{\left(\ell_{1}+p_{2}\right)^{2}}{p_{3}^{2}}\right)}}
 \,,
\end{align}
\label{eq:3pt_planar_ids}
\end{subequations}
\end{widetext}
where we have taken into account the expressions of leading singularities for 
one-loop triangle and box in four space-time dimensions provided by 
theorems~\ref{th:lan1} and~\ref{th:lan2}. 
Thus, obtaining the explicit integration 
of integrands containing square roots of K\"all\'en functions (whose arguments are propagators of the loop topology). 
A similar result is obtained for the application of the loop-by-loop approach in
the non-planar two-loop triangle of Fig.~\ref{fig:tri2L-b}.

From results~\eqref{eq:3pt_planar_ids}, we observe that the presence of square roots in the denominator
does not play any relevant role when computing the residue. 
We checked that these square roots can be rationalised and left the integrand 
in terms of only propagators. 
Therefore, we also obtain the four-dimensional leading singularity 
of a planar double box with massive external momenta and massless propagators 
(as depicted in Fig~\ref{fig:box2L-c}), 
\begin{widetext}
\begin{subequations}
\begin{align}
J_{\text{P;\ensuremath{4}-pt}}^{(2),D=4}=&\int\frac{d^{4}\ell_{1}}{\imath\pi^{2}}\frac{d^{4}\ell_{2}}{\imath\pi^{2}}\frac{1}{\ell_{1}^{2}\left(\ell_{1}-p_{1}\right)^{2}\left(\ell_{1}+p_{2}\right)^{2}\left(\ell_{1}-\ell_{2}\right)^{2}\left(\ell_{2}-p_{1}\right)^{2}\left(\ell_{2}+p_{2}\right)^{2}\left(\ell_{2}+p_{2}+p_{3}\right)^{2}}
\notag\\
=&\int\frac{d^{4}\ell_{2}}{\imath\pi^{2}}\frac{1}{\left(\ell_{2}-p_{1}\right)^{2}\left(\ell_{2}+p_{2}\right)^{2}\left(\ell_{2}+p_{2}+p_{3}\right)^{2}}\int\frac{d^{4}\ell_{1}}{\imath\pi^{2}}\frac{1}{\ell_{1}^{2}\left(\ell_{1}-p_{1}\right)^{2}\left(\ell_{1}+p_{2}\right)^{2}\left(\ell_{1}-\ell_{2}\right)^{2}}
\notag\\
\sim&\int\frac{d^{4}\ell_{1}}{\imath\pi^{2}}\frac{1}{\left(\ell_{2}-p_{1}\right)^{2}\left(\ell_{2}+p_{2}\right)^{2}\left(\ell_{2}+p_{2}+p_{3}\right)^{2}\ell_{2}^{2}s\sqrt{\lambda_{\text{K}}\left(1,\frac{p_{1}^{2}\left(\ell_{2}+p_{2}\right)^{2}}{s\ell_{2}^{2}},\frac{p_{2}^{2}\left(\ell_{2}-p_{2}\right)^{2}}{s\ell_{2}^{2}}\right)}}\,,
\end{align}
\label{eq:4pt_planar_ids}
where, with the observation we just made,
\begin{align}
J_{\text{P;\ensuremath{4}-pt}}^{(2),D=4} & \sim\frac{1}{s\sqrt{\lambda_{\text{K}}\left(1,\frac{p_{1}^{2}p_{3}^{2}}{st},\frac{p_{2}^{2}p_{4}^{2}}{st}\right)}}\,.
\end{align}
\end{subequations}
\end{widetext}
with $s=(p_1+p_2)^2$ and $t=(p_2+p_3)^2$. 

\paragraph{$L$-loop ladder triangle and box integrals\\}

With the results of Eqs.~\eqref{eq:3pt_planar_ids} and~\eqref{eq:4pt_planar_ids}
at our disposal we can provide explicit expressions
of leading singularities for ladder $L$-loop topologies, by exploiting the loop-by-loop approach. 
Our findings are summarised in the following theorems. 

\begin{theorem}
\label{th:2L_lead1}
The leading singularity of the three-point $L$-loop ladder Feynman integral of Fig.~\ref{fig:ladders-a} 
in four space-time dimensions with off-shell external momenta ($p_i^2\ne0$ for $i=1,2,3$) 
and massless propagators
is equal to 
$\left(\left(p_{3}^{2}\right)^{L}\sqrt{\lambda_{\text{K}}\left(1,\frac{p_{1}^{2}}{p_{3}^{2}},\frac{p_{2}^{2}}{p_{3}^{2}}\right)}\right)^{-1}$.
\end{theorem}
\begin{proof}
We prove this theorem by induction on $L$. 
For $L=2$, it was explicitly shown in Sec.~\ref{sec:leray2L} through the direct application of Leray's residues. 
Let us then consider the $L=3$ case. 

We begin with the three-loop planar integral, 
\begin{align}
J_{\text{P;\ensuremath{3}-pt}}^{(3),D=4}=\int\prod_{i=1}^{3}\frac{d^{4}\ell_{i}}{\imath\pi^{2}}\frac{1}{\left(\ell_{i}-p_{1}\right)^{2}\left(\ell_{i}+p_{2}\right)^{2}\left(\ell_{i-1}-\ell_{i}\right)^{2}}\,,
\end{align}
with $\ell_0=0$ (see Fig.~\ref{fig:ladders-a}). 

We then use the result of the leading singularity in $D=4$ of the two-loop ($L=2$) planar triangle, 
\begin{widetext}
\begin{align}
J_{\text{P;\ensuremath{3}-pt}}^{(3),D=4} & =\int\frac{d^{4}\ell_{3}}{\imath\pi^{2}}\frac{1}{\left(\ell_{3}-p_{1}\right)^{2}\left(\ell_{3}+p_{2}\right)^{2}}\int\left(\prod_{i=1}^{2}\frac{d^{4}\ell_{i}}{\imath\pi^{2}}\frac{1}{\left(\ell_{i}-p_{1}\right)^{2}\left(\ell_{i}+p_{2}\right)^{2}\left(\ell_{i-1}-\ell_{i}\right)^{2}}\right)\frac{1}{\left(\ell_{2}-\ell_{3}\right)^{2}}\nonumber \\
 & \sim\int\frac{d^{4}\ell_{3}}{\imath\pi^{2}}\frac{1}{\left(\ell_{3}-p_{1}\right)^{2}\left(\ell_{3}+p_{2}\right)^{2}\ell_{3}^{2}}\frac{1}{\left(p_{3}^{2}\right)^{2}\sqrt{\lambda_{\text{K}}\left(1,\frac{p_{1}^{2}\left(\ell_{3}+p_{2}\right)^{2}}{p_{3}^{2}\ell_{3}^{2}},\frac{p_{2}^{2}\left(\ell_{3}-p_{1}\right)^{2}}{p_{3}^{2}\ell_{3}^{2}}\right)}}\,,
\end{align}
\end{widetext}
where by applying the results~\eqref{eq:3pt_planar_ids}, we find, 
\begin{align}
J_{\text{P;\ensuremath{3}-pt}}^{(3),D=4} & \sim\frac{1}{\left(p_{3}^{2}\right)^{3}\sqrt{\lambda_{\text{K}}\left(1,\frac{p_{1}^{2}}{p_{3}^{2}},\frac{p_{2}^{2}}{p_{3}^{2}}\right)}}\,,
\end{align}
in agreement with the induction hypothesis. 

Let us now assume that the induction hypothesis is true for $L$ (with $L\geq3$). 
Then, with the same procedure we just followed for $L=3$, we find, 
\begin{widetext}
\begin{align}
&J_{\text{P;\ensuremath{3}-pt}}^{(L+1),D=4} =\int\left(\prod_{i=1}^{L+1}\frac{d^{4}\ell_{i}}{\imath\pi^{2}}\frac{1}{\left(\ell_{i}-p_{1}\right)^{2}\left(\ell_{i}+p_{2}\right)^{2}\left(\ell_{i-1}-\ell_{i}\right)^{2}}\right) \notag \\
 & =\int\frac{d^{4}\ell_{L+1}}{\imath\pi^{2}}\frac{1}{\left(\ell_{L+1}-p_{1}\right)^{2}\left(\ell_{L+1}+p_{2}\right)^{2}} \int\left(\prod_{i=1}^{L}\frac{d^{4}\ell_{i}}{\imath\pi^{2}}\frac{1}{\left(\ell_{i}-p_{1}\right)^{2}\left(\ell_{i}+p_{2}\right)^{2}\left(\ell_{i-1}-\ell_{i}\right)^{2}}\right) \frac{1}{\left(\ell_{L+1}-\ell_{L}\right)^{2}}\notag \\
 & \sim\int\frac{d^{4}\ell_{L+1}}{\imath\pi^{2}}\frac{1}{\left(\ell_{L+1}-p_{1}\right)^{2}\left(\ell_{L+1}+p_{2}\right)^{2}\ell_{L+1}^{2}}\frac{1}{\left(p_{3}^{2}\right)^{L}\sqrt{\lambda_{\text{K}}\left(1,\frac{p_{1}^{2}\left(\ell_{L+1}+p_{2}\right)^{2}}{p_{3}^{2}\ell_{L+1}^{2}},\frac{p_{2}^{2}\left(\ell_{L+1}-p_{1}\right)^{2}}{p_{3}^{2}\ell_{L+1}^{2}}\right)}}\,,
 \notag
\end{align}
\end{widetext}
proving in this way our theorem,
\begin{align}
J_{\text{P;\ensuremath{3}-pt}}^{(L+1),D=4} & \sim\frac{1}{\left(p_{3}^{2}\right)^{L+1}\sqrt{\lambda_{\text{K}}\left(1,\frac{p_{1}^{2}}{p_{3}^{2}},\frac{p_{2}^{2}}{p_{3}^{2}}\right)}}\,.
\end{align}
\end{proof}

\begin{theorem}
\label{th:2L_lead2}
The leading singularity of the four-point $L$-loop ladder Feynman integral of Fig.~\ref{fig:ladders-b} 
in four space-time dimensions with off-shell external momenta ($p_i^2\ne0$ for $i=1,2,3,4$) 
and massless propagators
is equal to 
$\left(s^{L}t\sqrt{\lambda_{\text{K}}\left(1,\frac{p_{1}^{2}p_{3}^{2}}{st},\frac{p_{2}^{2}p_{4}^{2}}{st}\right)}\right)^{-1}$,
with $s=(p_1+p_2)^2$ and $t=(p_2+p_3)^2$.
\end{theorem}
\begin{proof}
Similar to theorem~\ref{th:2L_lead1}, the proof is carried out by induction on $L$. 
For $L=2$, it was already computed in Eqs.~\eqref{eq:4pt_planar_ids}.
For $L=3$, let us start considering the three-loop Feynman integral,
\begin{align}
J_{\text{P;\ensuremath{4}-pt}}^{(3),D=4} &=\int\left(\prod_{i=1}^{3}\frac{d^{4}\ell_{i}}{\imath\pi^{2}}\frac{1}{\left(\ell_{i}-p_{1}\right)^{2}\left(\ell_{i}+p_{2}\right)^{2}\left(\ell_{i-1}-\ell_{i}\right)^{2}}\right) \notag \\
&\times \frac{1}{\left(\ell_{3}+p_{23}\right)^{2}}\,,
\end{align}
then, by using the results~\eqref{eq:4pt_planar_ids}
for the two-loop ($L=2$) planar double box, 
\begin{widetext}
\begin{align}
J_{\text{P;\ensuremath{4}-pt}}^{(3),D=4} & =\int\frac{d^{4}\ell_{3}}{\imath\pi^{2}}\frac{1}{\left(\ell_{3}-p_{1}\right)^{2}\left(\ell_{3}+p_{2}\right)^{2}\left(\ell_{3}+p_{23}\right)^{2}}\int\left(\prod_{i=1}^{2}\frac{d^{4}\ell_{i}}{\imath\pi^{2}}\frac{1}{\left(\ell_{i}-p_{1}\right)^{2}\left(\ell_{i}+p_{2}\right)^{2}\left(\ell_{i-1}-\ell_{i}\right)^{2}}\right)
\frac{1}{\left(\ell_2-\ell_3\right)^2}
\nonumber \\
 & \sim\int\frac{d^{4}\ell_{3}}{\imath\pi^{2}}\frac{1}{\left(\ell_{3}-p_{1}\right)^{2}\left(\ell_{3}+p_{2}\right)^{2}\left(\ell_{3}+p_{23}\right)^{2}\ell_{3}^{2}s^{2}\sqrt{\lambda_{\text{K}}\left(1,\frac{p_{1}^{2}\left(\ell_{3}+p_{2}\right)^{2}}{s\ell_{3}^{2}},\frac{p_{2}^{2}\left(\ell_{3}-p_{1}\right)^{2}}{s\ell_{3}^{2}}\right)}}
\end{align}
\end{widetext}
we find agreement with the induction hypothesis.
\begin{align}
J_{\text{P;\ensuremath{4}-pt}}^{(3),D=4} & \sim\frac{1}{s^{3}t\sqrt{\lambda_{\text{K}}\left(1,\frac{p_{1}^{2}p_{3}^{2}}{st},\frac{p_{2}^{2}p_{4}^{2}}{st}\right)}}\,.
\end{align}
After explicitly showing $L=2$ and $L=3$ cases, we can assume that our induction hypothesis is true
for $L$ (with $L\geq3$). Then, as similarly carried out in the previous theorem and 
making use of results~\eqref{eq:4pt_planar_ids}, we find, 
\begin{widetext}
\begin{align}
&J_{\text{P;\ensuremath{4}-pt}}^{(L+1),D=4} = \notag \\ &=\int\frac{d^{4}\ell_{L+1}}{\imath\pi^{2}}\frac{1}{\left(\ell_{L+1}-p_{1}\right)^{2}\left(\ell_{L+1}+p_{2}\right)^{2}\left(\ell_{L+1}+p_{2}+p_{3}\right)^{2}}\int\left(\prod_{i=1}^{L}\frac{d^{4}\ell_{i}}{\imath\pi^{2}}\frac{1}{\left(\ell_{i}-p_{1}\right)^{2}\left(\ell_{i}+p_{2}\right)^{2}\left(\ell_{i-1}-\ell_{i}\right)^{2}}\right)
 \frac{1}{\left(\ell_L-\ell_{L+1}\right)^2}
 \nonumber \\
 & \sim\int\frac{d^{4}\ell_{L+1}}{\imath\pi^{2}}\frac{1}{\left(\ell_{L+1}-p_{1}\right)^{2}\left(\ell_{L+1}+p_{2}\right)^{2}\left(\ell_{L+1}+p_{2}+p_{3}\right)^{2}\ell_{L+1}^{2}}\frac{1}{s^{L}\sqrt{\lambda_{\text{K}}\left(1,\frac{p_{1}^{2}\left(\ell_{L+1}+p_{2}\right)^{2}}{s\ell_{L+1}^{2}},\frac{p_{2}^{2}\left(\ell_{L+1}-p_{1}\right)^{2}}{s\ell_{L+1}^{2}}\right)}}\nonumber \\
 & \sim\frac{1}{s^{L+1}t\sqrt{\lambda_{\text{K}}\left(1,\frac{p_{1}^{2}p_{3}^{2}}{st},\frac{p_{2}^{2}p_{4}^{2}}{st}\right)}}\,.
\end{align}
\end{widetext}
Thus, completing the proof of this theorem. 
\end{proof}

\par\bigskip
Interestingly, from the results of theorems~\ref{th:2L_lead1} and~\ref{th:2L_lead2} 
and their extension to higher multiplicity in external momenta and non-planar topologies, 
it is possible to predict $d\log$ integrands at multi-loop level without an explicit calculation. 
This information turns out to be very useful in techniques that aim at the construction of differential equation
in canonical form from the knowledge of a single integral, e.g. Ref.~\cite{Dlapa:2020cwj}. 

Also, we would like to remark that the results of theorems~\ref{th:2L_lead1} and~\ref{th:2L_lead2}
were found in Ref.~\cite{Usyukina:1993ch} by following an approach based on Mellin-Barnes 
representation of Feynman integrals.

\subsection{Observations on differential equations in canonical form}

In view of the loop-by-loop approach considered in Sec.~\ref{sec:loopbyloop},
obtained from leading and Landau singularities of Feynman integrals
at one-loop, we can combine this approach with evaluation of integrals
that are needed for the calculation of scattering amplitudes at high
energies. Previously, we remarked that the knowledge of $d\log$ integrals
in the top sector in interplay with the algorithm~\cite{Dlapa:2020cwj} can automatically
lead us to construct a canonical basis of master integrals. In this
section, we plan to avoid the use of the latter algorithm and profit, instead, from the study of multi-loop Feynman
integrals with constant leading singularity in a given dimension. 
To complete the canonical basis, we use the knowledge of leading singularities
in different dimensions, extending the discussion of Sec.~\ref{sec:deq1L}. 

\begin{figure}[t]
    \centering
    \includegraphics[scale=0.7]{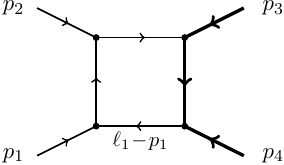}\qquad
    \includegraphics[scale=0.6]{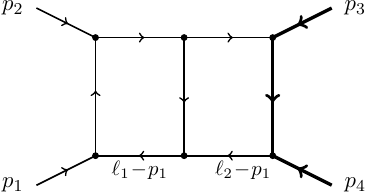}
    \caption{One- and two-loop integral families. Thin (thick) lines correspond to massless (massive) propagators}
    \label{fig:2L_parent}
\end{figure}

To start, let us draw our attention to the integral families depicted
in Fig.~\ref{fig:2L_parent}. These integrals appear in the calculation of one- and
two-loop scattering amplitudes for heavy-quark pair production through
parton-parton annihilation in QCD. We can then express these families,
respectively, as, 
\begin{widetext}
\begin{align}
J_{4}^{\left(1\right),D}\left(1^{n_{1}},2^{n_{2}},3^{n_{3}},4^{n_{4}}\right)= & \int\frac{d^{D}\ell_{1}}{\imath\,\pi^{D/2}}\frac{1}{\left(\ell_{1}-p_{1}\right)^{2n_{1}}\ell_{1}^{2n_{2}}\left(\ell_{2}+p_{2}\right)^{2n_{3}}\left(\left(\ell_{2}+p_{23}\right)^{2}-m^{2}\right)^{n_{4}}}\,,\\
J_{4}^{\left(1\right),D}\left(1^{n_{1}},\ldots,7^{n_{7}};8^{n_{8}},9^{n_9}\right)= & \int\frac{d^{D}\ell_{1}}{\imath\,\pi^{D/2}}\frac{d^{D}\ell_{2}}{\imath\,\pi^{D/2}}\frac{1}{\left(\ell_{1}-p_{1}\right)^{2n_{1}}\ell_{1}^{2n_{2}}\left(\ell_{1}+p_{2}\right)^{2n_{3}}}\nonumber \\
 & \times\frac{\left(\left(\ell_{1}+p_{23}\right)^{2}-m_{2}^{2}\right)^{-n_{8}}\left(\ell_{2}^{2}\right)^{-n_{9}}}{\left(\ell_{2}-p_{1}\right)^{2n_{4}}\left(\ell_{2}+p_{2}\right)^{2n_{5}}\left(\left(\ell_{2}+p_{23}\right)^{2}-m^{2}\right)^{n_{6}}\left(\ell_{1}-\ell_{2}\right)^{2n_{7}}}\,,
\end{align}
\end{widetext}
with $p_1^2=p_2^2=0$, $p_3^2=p_4^2=m^2$, the momentum conservation $p_1+p_2+p_3+p_4=0$,
and the kinematic invariants, $s=(p_1+p_2)^2, t=(p_2+p_4)^2$. 
We follow the convention of~\eqref{eq:intL}. These integral families are known to have,
respectively, 5 and 21 master integrals. While at one-loop it is straightforward
to obtain a canonical basis from the knowledge of leading and Landau
singularities in different dimensions (see Sec.~\ref{sec:deq1L}), for the two-loop
case, instead, a loop-by-loop approach can be considered. We follow
an alternative approach based solely on the study of $d\log$ integrals
in a specific dimension. Without loss of generality, we work on $D=4-2\epsilon$. 

As considered throughout this work, looking for $d\log$ integrals relies
on the loop-momentum parametrisation of the integrand. Thus, by following
the approach of~\cite{Henn:2020lye}, one can construct an ansatz of possible
candidates of $d\log$ integrals to then solve for basis of integrals
that only contain a constant leading singularity. This method has
been successfully developed in the code {\sc DlogBasis} by profiting from
the spinor-helicity formalism for the parametrisation of the loop
momenta. 

\begin{figure}[t]
    \centering
    \includegraphics[scale=0.6]{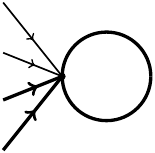}\quad
    \includegraphics[scale=0.6]{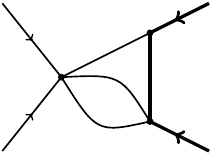}\,
    \includegraphics[scale=0.6]{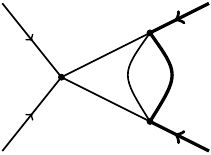}\,
    \includegraphics[scale=0.6]{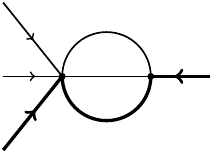}
    \includegraphics[scale=0.6]{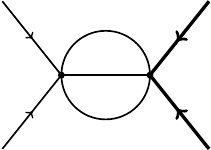}\,
    \includegraphics[scale=0.6]{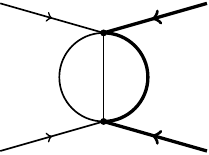}\,
    \includegraphics[scale=0.6]{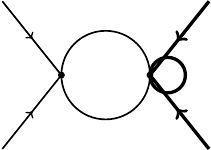}
    \caption{Missing canonical master integrals.}
    \label{fig:2L_missing}
\end{figure}

Owing to the {\sc DlogBasis} implementation, we compute candidates for the
one- and two-loop integral families of Fig.~\ref{fig:2L_parent}. We realise that, after
IBP reduction, there are missing integrals in the sectors displayed
in Fig.~\ref{fig:2L_missing}. These integrals appear in subsectors and only contain
insertions of bubble and tadpole loops. 
We consider these sub-loop integral in $D=2-2\epsilon$,
whilst the other loop remains in $D=4-2\epsilon$. Let us look at this
in detail.

For the integrals displayed in Fig.~\ref{fig:2L_missing}, we find the canonical integrals
to admit the following form, 
\begin{align*}
J_{4}^{\left(1\right),D=4-2\epsilon}\left(4\right) & \to J_{4}^{\left(1\right),D=2-2\epsilon}\left(4\right)\,,\\
J_{4}^{\left(2\right),D=4-2\epsilon}\left(3,4,6,7\right) & \to r_\lambda\,J_{4}^{\left(2\right),D=4-2\epsilon}\left(3,4,6,7^{2}\right)\,,\\
J_{4}^{\left(2\right),D=4-2\epsilon}\left(1,3,6,7\right) & \to s\,J_{4}^{\left(2\right),D=2-2\epsilon}\left(1,3,6,7;8^{-1}\right)\,,\\
J_{4}^{\left(2\right),D=4-2\epsilon}\left(3,6,7\right) & \to m^{2}\,J_{4}^{\left(2\right),D=4-2\epsilon}\left(3,6,7;8^{-1}\right)\,,\\
J_{4}^{\left(2\right),D=4-2\epsilon}\left(3,4,7\right) & \to s\,J_{4}^{\left(2\right),D=2-2\epsilon}\left(3,4,7\right)\\
J_{4}^{\left(2\right),D=4-2\epsilon}\left(2,6,7\right) & \to t\,J_{4}^{\left(2\right),D=2-2\epsilon}\left(2,6,7;8^{-1}\right)\\
J_{4}^{\left(2\right),D=4-2\epsilon}\left(1,3,6\right) & \to s\,J_{4}^{\left(2\right),D=2-2\epsilon}\left(1,3,6\right)\,,
\end{align*}
with $r_\lambda=\sqrt{\lambda_{\text{K}}\left(s,m^{2},m^{2}\right)}=\sqrt{-s\left(-s+4m^{2}\right)}$.
Integrals that only contain a bubble sub-loop are independently
treated in $D=2-2\epsilon$ and then shifted to $D=4-2\epsilon$.
Thus, introducing squared propagators. For the canonical integrals
that contain an insertion of a numerator, we directly consider them
in $D=2-2\epsilon$. With this additional set of master integrals,
we recover the canonical differential equations. 

\section{Conclusions}

Elaborating on the decomposition of the space-time dimension into two independent and complementary subspaces, parallel and perpendicular, in this paper, we carried out a study of the analytic properties of Feynman integrals within this framework. We observed that the various features of Feynman integrals can systematically be tracked down to an analysis at the integrand level, in which techniques and results based on integrand reduction methods were crucial for the results presented in this paper. 

Since it is known that Feynman integrals display singularities depending on kinematic
invariants in given space-time dimensions, we considered Landau and leading singularities as well
as the connections between them. 
For Landau singularities, because of the way how the space-time dimension was decomposed, 
we provided and proved the algorithm to calculate these singularities 
by focusing only on a linear system of equations.  
We observed that Landau singularities can be cast in the single variable \eqref{lambda_landau}, present
in the loop-momentum parametrization, whose value is determined by imposing
on-shell conditions. 

Regarding leading singularities, 
we observed for one-loop scalar Feynman integrals that the number of integrations is lessened to the number of independent external momenta ($E$). This feature, obtained as a product of the general decomposition of the space-time dimension in terms of parallel and perpendicular directions, 
led us to perform an integrand analysis, finding a connection between leading and Landau singularities. In detail, we noted and proved that leading singularities in $D=E$ and $D=E+1$, respectively, correspond to the inverse of the square root of the leading Landau singularity of the first and second type. 
At the beginning we observed this connection with the explicit calculation of leading and Landau singularities of up-to six-point one-loop scalar Feynman integrals.
Then, to prove at all multiplicities, we use the multi-dimensional theory of residues provided by Leray. 

As a by-product of the connection between Landau and leading singularities, 
we observed that one-loop scalar Feynman integrals can be cast in a very simple product of $d\log$ forms. 
We checked this for integrals up to six points and stated the conjecture that it holds 
for all one-loop scalar Feynman integrals.
Additionally, we briefly discussed the connection of our work with the method of differential equations in dimensional regularisation. 
In this matter, the knowledge of leading singularities for an arbitrary one-loop graph 
makes the problem of finding the canonical form for a differential equation much simpler.

On top of the aforementioned one-loop analysis, 
we obtained, promising preliminary results for Feynman multi-loop integrals by applying developed methods. 
We considered particular families of integrals, namely $L$-loop three- and four-point ladder diagrams. 
Treating them as iterated integrals, we were able to apply Leray's residues in the loop-by-loop approach to establish their leading singularities. 
Continuing on the multi-loop application of the method we discussed the canonical differential equation for heavy-quark pair production through parton-parton annihilation in QCD at the two-loop level.

Our work which was focused mostly on the one-loop Feynman integrals just scratch the surface of possible applications of methods presented here and other ideas coming from mathematics. We were able to identify possible directions for further studies
\begin{enumerate}
\item The first natural extension of this work is to investigate the application of multi-dimensional residues to general multi-loop Feynman integrals, starting from two loops, and its connection to the notion of leading singularity. 
At two loops, one can already observe that the number of poles is rarely equal to the degree of the form, so even by taking the global residue, we cannot reduce this form to the zero form, which means that there is some integration which is left. Here the crucial point will be the determination of cycles of integration. In the one-loop case, this problem was already studied in~\cite{Abreu:2017ptx} in the context of the calculation of non-maximal residues. 

\item Further investigation of the connection between the differential equation method, and Landau singularities.
On the one hand, it would be interesting to see how much we can learn from the knowledge of leading singularities to find the canonical form for differential equations for Feynman integrals
as well as using the knowledge of Landau singularities in integrals that are known to be described 
by elliptic curved and beyond. 
On the other hand, it would be very interesting to see if one can recover the full alphabet from Landau singularities by considering all possible residues for a given diagram. In that context, it would also be interesting to look at the connection to cluster algebras.

\item Lastly, it would be very interesting to investigate the $d\log$ structure of Feynman integrals. 
The multi-dimensional residues can shed new light on this issue.
It would be ideal to find a set of criteria which determine whether a given differential form has a full $d\log$ structure. In connection to this, it also would be interesting to see if it is possible to determine the geometry necessary for a given Feynman diagram by considering its Landau singularities.
    
\end{enumerate}

\noindent All the above perspectives can be explored within the approach proposed in this work.

\section*{Acknowledgments}

We are indebted to Johannes Henn for discussions and feedback on the manuscript,
and Janusz Gluza and Einan Gardi for comments on the manuscript. 
We also wish to thank Jungwon Lim, Sebastian Mizera, and Andrew McLeod for insightful discussions.
We are thankful to the anonymous Referee for carefully reading our manuscript and for providing useful comments that led to an improved version of the results summarised in this communication. 
This research received funding from the European Research Council (ERC) under the European Union's Horizon 2020 research and innovation programme (grant agreement No 725110), {\it Novel structures in scattering amplitudes},
the Polish National Science Center (NCN) under grant 2023/50/A/ST2/00224,
and by the Leverhulme Trust, LIP-2021-01.
This research was supported by the Munich Institute for Astro-, Particle and BioPhysics (MIAPbP) which is funded by the Deutsche Forschungsgemeinschaft (DFG, German Research Foundation) under Germany's Excellence Strategy -- EXC-2094 -- 390783311.

\appendix

\section{Leray residues}\label{app_leray}
We will give a brief introduction, without any proofs, to the multi-dimensional theory of residues due to Leray. 
It is one of the possible extensions of the one-dimensional residue theory. 
It is particularly useful in our application as it introduces a residue operator that acts directly on differential forms. 
For a detailed presentation and proofs of this theory, we refer the Reader to the original work of Leray~\cite{Leray_1959} or \cite{Pham2011full,Aizenberg1994}.

Let $X$ be a complex analytic manifold and $S$ a closed analytic submanifold of complex codimension 1.
If $\phi$ is a closed differential form in $X-S$ with a pole of first order on $S$, then in a neighbourhood $U_{a}$
of any point $a \in S$ the form $\phi$ can be represented as,
\begin{equation}
\phi = \frac{ds}{s}\wedge \psi + \theta\,,
\end{equation}
where, $s=s_{a}(z)$, is a defining function for the manifold $S$ in $U_{a}$, 
and $\psi$, $\theta$ are forms which are regular on $U_{a}$. The restriction $\psi \vert_{S}$ is called the residue form of the form $\phi$, $\psi \vert_{S} \equiv res[\phi]$.
\begin{theorem}
For an arbitrary closed form $\phi$ of degree $p$ on $X-S$ and a cycle $\sigma \in Z_{p-1}(S)$ there is a formula, 
\begin{equation}
\int_{\delta \sigma} \phi = 2 \pi i \int_{\sigma} res[\phi]\,,    
\end{equation}
where $\delta \sigma \in Z_{p}(X-S)$ is the co-boundary of the cycle $\sigma$.
\label{th:residue}
\end{theorem}
If the form $\phi \in Z^{p}(X-(S_{1} \cup \hdots \cup S_{m}))$ has a pole of first order on $S_{1}, \hdots , S_{m}$ then one can define an iterated residue form $res^{m}[\phi] \in Z^{p-m}(S_{1} \cap \hdots \cap S_{m})$ and a homomorphism,
\begin{equation}
res^{m}: H^{p}(X-(S_{1} \cup \hdots \cup S_{m})) \rightarrow H^{p-m}(S_{1} \cap \hdots \cap S_{m})\,,
\end{equation}
as the composition of homomorphisms,
\begin{equation}
\begin{split}
&H^{p}(X-(S_{1} \cup \hdots \cup S_{m})) \xrightarrow{res} H^{p-1}(S_{1}-(S_{2} \cup \hdots \cup S_{m})) \\
&\xrightarrow{res} \hdots \xrightarrow{res} H^{p-m}(S_{1} \cap \hdots \cap S_{m})\,.
\end{split}
\end{equation}
There is also a sequence of co-boundary homomorphisms
\begin{equation}
\begin{split}
&H_{p}(X-(S_{1} \cup \hdots \cup S_{m})) \xleftarrow{\delta_{1}} H_{p-1}(S_{1}-(S_{2} \cup \hdots \cup S_{m})) \\
&\xleftarrow{\delta_{2}} \hdots \xleftarrow{\delta_{m}} H_{p-m}(S_{1} \cap \hdots \cap S_{m})\,.  
\end{split}
\end{equation}
Iterating Theorem \ref{th:residue} we can write the formula for composed residue
\begin{equation}
\int_{\delta^{m} \sigma} \phi = (2 \pi i)^{m} \int_{\sigma} res^{m}[\phi]\,.    
\end{equation}

The composed residue of a form $\phi$ with polar singularities of order $q_{1}, q_{2},\hdots , q_{m}$ on $S_{1}, \hdots, S_{m}$ is written
\begin{equation}
\begin{split}
&res^{m}[\phi] = \\
&=\frac{1}{(q_{1}-1)!\hdots (q_{m}-1)!}\frac{d^{q_{1}+\hdots + q_{m} -m} \omega}{dS_{1}^{q_{1}} \wedge \hdots \wedge dS_{m}^{q_{m}}}\vert_{S_{1} \cap \hdots \cap S_{m}}\,,
\label{comp_res}
\end{split}
\end{equation}
where $\omega = s_{1}^{q_{1}}\hdots s_{m}^{q_{m}}\phi$.

In our case, we deal with scalar one-loop Feynman integrals that in momentum components
can be cast as, 
\begin{equation}
\phi = \frac{dk_{1} \wedge \hdots \wedge dk_{D}}{D_{1}\hdots D_{m}}\,,
\end{equation}
where $D_{i} = (k+P_{i-1})^{2}-m_{i}^{2}$, with $P_{i-1}=p_1+\hdots+p_{i-1}$ and $P_0=0$.

\begin{prop}
A one-loop scalar Feynman integrand $\phi = \frac{dk_{1} \wedge \hdots \wedge dk_{D}}{D_{1}\hdots D_{m}} = f dk_{1} \wedge \hdots \wedge dk_{d}$ is a closed differential form, i.e., $d \phi = 0$.
\end{prop}
\begin{proof}
We are in a $D$-dimensional $k$-space, i.e., $\phi$ is a top holomorphic form. 
Therefore, $d \phi = (\frac{\partial f}{\partial k_{1}}dk_{1} + \hdots \frac{\partial f}{\partial k_{D}}dk_{D}) \wedge dk_{1} \wedge \hdots \wedge dk_{D} = 0$. 
\end{proof}

Since we are interested in the calculation of composite residue form for Feynman integrands, 
we have $\omega = \frac{D_{1}\hdots D_{m} dk_{1} \wedge \hdots \wedge dk_{D}}{D_{1}\hdots D_{m}} = dk_{1} \wedge \hdots \wedge dk_{D}$ and $q_{1}= \hdots = q_{m} = 1 \Rightarrow q_{1} + \hdots + q_{m} =m$.
Then, we need to consider
the highest co-dimensional residue, i.e., all $D_{i}$ intersect $D_{1}=\hdots = D_{m} = 0$,
and the case when $d=m$, thus,
\begin{equation}
res^{n}[\phi] = \frac{dz_{1}\wedge \hdots \wedge dz_{n}}{dD_{1} \wedge \hdots \wedge dD_{n}} \in H^{0}(D_{1} \cap \hdots \cap D_{n})\,.
\end{equation}
\paragraph{Example: one-loop bubble in $D=2$}
\begin{equation}
\begin{split}
&\phi = \frac{dk_{1} \wedge dk_{2}}{(k^2-m^{2})((k+p)^{2}-m^{2})}\,, \\
&D_{1} = k^2 - m^{2} = k_{1}^{2} - k_{2}^{2} - m^{2}\,, \\
&D_{2} = (k+p)^{2}-m^{2} = (k_{1}+p_{1})^{2} - (k_{2}+p_{2})^{2} - m^{2}\,.
\end{split}
\end{equation}
The differentials of polar sets $D_{1}$ and $D_{2}$ are,
\begin{equation}
\begin{split}
&dD_{1} = 2k_{1}dk_{1} - 2k_{2}dk_{2}\,, \\
&dD_{2} = 2(k_{1}+p_{1})dk_{1} - 2(k_{2}+p_{2})dk_{2}\,,
\end{split}
\end{equation}
and their wedge product is
\begin{align}
&dD_{1} \wedge dD_{2} = \notag \\
&= -4k_{1}(k_{2}+p_{2})dk_{1}\wedge dk_{2} - 4k_{2}(k_{1}+p_{1})dk_{2}\wedge dk_{1}
\notag\\ 
&= -4(k_{1}p_{2} - k_{2}p_{1})dk_{1}\wedge dk_{2}\,.
\end{align}
Hence the residue form is given by
\begin{equation}
\begin{split}
&res^{2}[\phi] = \\
&=\frac{ dk_{1} \wedge dk_{2}}{-4(k_{1}p_{2} - k_{2}p_{1})dk_{1}\wedge dk_{2}} =  \frac{ 1}{-4(k_{1}p_{2} - k_{2}p_{1})}\,.
\end{split}
\end{equation}
By solving $D_{1}=D_{2}=0$ for $k_{1}$ and $k_{2}$ we get two solutions which give us
\begin{equation}
res^{2}[\phi] = \frac{\pm 1}{2\sqrt{-s(4m^{2}-s)}},    
\end{equation}
which is the expected result.

\section{Supplementary results}\label{app_supp}
In this section, we collect results which we use to prove statements in the main text.

\begin{lemma}
The sum $\sum_{k=1}^{n} \det(A_{1}A_{2}\hdots I_{k}\hdots A_{n})$ is equal to $\det(I(A_{2}-A_{1})(A_{3}-A_{1})\hdots (A_{n}-A_{1}))$, where $A_{j}=(a_{1j},...,a_{nj})^{T}$ for $j=1,\hdots,n$ are columns of the determinants and $I_{k}=(1,\hdots,1)^{T}$ replaces the k'th column $A_{k}$ in k'th term of the sum.
\label{sum_det}
\end{lemma}
\begin{proof}
\begin{equation}
\begin{split}
&\vert I A_{2}A_{3}\hdots A_{n} \vert + \vert A_{1}IA_{3}\hdots A_{n} \vert +\\
&+ \hdots + \vert A_{1}A_{2}A_{3}\hdots A_{n-1}I \vert = \\
&\vert I A_{2}A_{3}\hdots A_{n} \vert - \vert IA_{1}A_{3}\hdots A_{n} \vert \\
&+ \hdots + \vert A_{1}A_{2}A_{3}\hdots A_{n-1}I \vert = \\
&\vert I (A_{2}-A_{1})A_{3}\hdots A_{n} \vert + \vert A_{1}A_{2}I\hdots A_{n} \vert + \\
&+ \hdots + \vert A_{1}A_{2}A_{3}\hdots A_{n-1}I \vert = \\ 
&\vert I (A_{2}-A_{1})A_{3}\hdots A_{n} \vert - \vert IA_{2}A_{1}\hdots A_{n} \vert + \\
&+ \hdots + \vert A_{1}A_{2}A_{3}\hdots A_{n-1}I \vert = \\
&\vert I (A_{2}-A_{1})A_{3}\hdots A_{n} \vert - \vert I(A_{2}-A_{1})A_{1}\hdots A_{n} \vert + \\
&+ \hdots + \vert A_{1}A_{2}A_{3}\hdots A_{n-1}I \vert = \\
&\vert I (A_{2}-A_{1})(A_{3}-A_{1})\hdots A_{n} \vert + \hdots + \vert A_{1}A_{2}A_{3}\hdots A_{n-1}I \vert
\end{split}
\end{equation}
By repeating the same procedure for remaining terms we get the result.
\end{proof}

\begin{lemma}
$\det(a_{ij})dx_{1} \wedge \hdots \wedge dx_{n} = (a_{11}dx_{1} + \hdots a_{1n}dx_{n}) \wedge \hdots \wedge (a_{n1}dx_{1} + \hdots a_{nn}dx_{n})$.
\label{wedge_det}
\end{lemma}
\begin{proof}
The statement will be proved by the induction. Let us prove this for $n=2$,
\begin{equation}
\begin{split}
&(a_{11}dx_{1} + a_{12}dx_{2})\wedge ( a_{21}dx_{1} + a_{22}dx_{2}) =
\\ 
&=a_{11}a_{21}dx_{1}\wedge dx_{1} + a_{11}a_{22}dx_{1} \wedge dx_{2} + \\
&+ a_{12}a_{21}dx_{2} \wedge dx_{1} + a_{12}a_{22}dx_{2} \wedge dx_{2} =  
\\
&= (a_{11}a_{22} - a_{12}a_{21})dx_{1} \wedge dx_{2}  
\\
&=
\begin{vmatrix}
a_{11} & a_{12} \\
a_{21} & a_{22}
\end{vmatrix}
dx_{1} \wedge dx_{2}\,.
\end{split}
\end{equation}
Let us assume that the statement holds for $n=m>2$, i.e.,
\begin{equation}
\begin{split}
&\begin{vmatrix}
a_{1i_{1}} & \hdots & a_{1i_{m}} \\
\vdots & \ddots & \vdots \\
a_{mi_{1}} & \hdots & a_{mi_{m}}
\end{vmatrix} dx_{i_{1}} \wedge \hdots \wedge dx_{i_{m}} = \\
&=(a_{1i_{1}}dx_{i_{1}} \wedge \hdots \wedge a_{1i_{m}}) \wedge \hdots \wedge (a_{mi_{1}}dx_{i_{1}} \wedge \hdots \wedge a_{mi_{m}})\,,
\end{split}
\end{equation}
with $i_{1} < \hdots <i_{m}$.

Let us now check the statement for $m+1$, 
\begin{equation}
\begin{vmatrix}
a_{11} & \hdots & a_{1m} & a_{1(m+1)} \\
\vdots & \ddots & \vdots & \vdots \\
a_{m1} & \hdots & a_{mm} & a_{m(m+1)} \\
a_{(m+1)1} & \hdots & a_{(m+1)m} & a_{(m+1)(m+1)}
\end{vmatrix} dx_{1} \wedge \hdots \wedge dx_{m} \wedge dx_{m+1}\,.   
\end{equation}
By expanding this determinant with respect to the last row, 
\begin{equation}
\begin{split}
&\Bigg(
(-1)^{(m+1)+1}a_{(m+1)1}
\begin{vmatrix}
a_{12} & \hdots & a_{1m} & a_{1(m+1)} \\
\vdots & \ddots & \vdots & \vdots \\
a_{m2} & \hdots & a_{mm} & a_{m(m+1)} 
\end{vmatrix}
+ \hdots + \\
&(-1)^{(m+1)+(m+1)}a_{(m+1)(m+1)}
\begin{vmatrix}
a_{12} & \hdots & a_{1m} \\
\vdots & \ddots & \vdots \\
a_{m2} & \hdots & a_{mm} 
\end{vmatrix}
\Bigg) \times\\
&dx_{1} \wedge \hdots \wedge dx_{m} \wedge dx_{m+1} = (-1)^{(m+1)+1}(-1)^{m} \times \\
&\left( \begin{vmatrix}
a_{12} & \hdots & a_{1m} & a_{1(m+1)} \\
\vdots & \ddots & \vdots & \vdots \\
a_{m2} & \hdots & a_{mm} & a_{m(m+1)} 
\end{vmatrix}dx_{2} \wedge \hdots \wedge dx_{m} \wedge dx_{m+1}\right)\wedge \\
&\wedge \left(a_{(m+1)1}dx_{1}\right) + \hdots +(-1)^{(m+1)+(m+1)} \times \\
&\left(
\begin{vmatrix}
a_{12} & \hdots & a_{1m} \\
\vdots & \ddots & \vdots \\
a_{m2} & \hdots & a_{mm} 
\end{vmatrix}
dx_{1} \wedge \hdots \wedge dx_{m}
\right) \wedge \left(a_{(m+1)(m+1)}dx_{m+1}\right)\,,
\end{split}
\end{equation}
where by induction hypothesis we get, 
\begin{equation}
\begin{split}
&(-1)^{2(m+1)}\big[(a_{12}dx_{2}+\hdots a_{1(m+1)}dx_{m+1})\wedge \hdots \wedge \\
&(a_{m2}dx_{2}+ \hdots +a_{m(m+1)}dx_{m+1}) \big] \wedge a_{(m+1)1}dx_{1} \\
&+ \hdots + \\
&(-1)^{2(m+1)}\big[(a_{11}dx_{1}+\hdots a_{1m}dx_{m})\wedge \hdots \wedge \\
&(a_{m1}dx_{1}+ \hdots +a_{mm}dx_{m}) \big] \wedge a_{(m+1)(m+1)}dx_{m+1}\,.
\end{split}
\end{equation}
However, because of $dx_{i} \wedge dx_{i}=0$, we can add to each bracket one missing one-form, 
\begin{equation}
\begin{split}
&[(a_{11}dx_{1} + a_{12}dx_{2}+\hdots a_{1(m+1)}dx_{m+1})
\wedge \hdots \wedge \\
& (a_{m1}dx_{1} + a_{m2}dx_{2}+ \hdots +a_{m(m+1)}dx_{m+1}) ] \wedge a_{(m+1)1}dx_{1}  \\
&+ \hdots + \\
&[(a_{11}dx_{1}+\hdots a_{1m}dx_{m}+a_{1(m+1)}dx_{m+1})\wedge \hdots \wedge \\
& (a_{m1}dx_{1}+ \hdots +a_{mm}dx_{m}+a_{m(m+1)}dx_{m+1}) ] \wedge \\ &\wedge a_{(m+1)(m+1)}dx_{m+1}\,,
\end{split}
\end{equation}
that allows us to factor out the square bracket, 
\begin{equation}
\begin{split}
&[(a_{11}dx_{1} + a_{12}dx_{2}+\hdots a_{1(m+1)}dx_{m+1})
\wedge \hdots \wedge \\
&(a_{m1}dx_{1} + a_{m2}dx_{2}+ \hdots +a_{m(m+1)}dx_{m+1}) ] \wedge 
\\
&\wedge( a_{(m+1)1}dx_{1} + \hdots + a_{(m+1)(m+1)}dx_{m+1} ) = \\
&=(a_{11}dx_{1} + a_{12}dx_{2}+\hdots a_{1(m+1)}dx_{m+1})\wedge \hdots \wedge \\
&(a_{m1}dx_{1} + a_{m2}dx_{2}+ \hdots +a_{m(m+1)}dx_{m+1}) \wedge \\
&\wedge( a_{(m+1)1}dx_{1} + \hdots + a_{(m+1)(m+1)}dx_{m+1} )\,,
\end{split}
\end{equation}
completing the proof of this lemma.
\end{proof}

Alternatively the statement of the Lemma \ref{wedge_det} can be seen as the definition of the determinant as follows:
Let $x'_i = \sum_{j =
1}^n a_{i j} x_j$.  Then\footnote{We thank the anonymous Referee for suggesting this alternative proof.}
\begin{align}
&\bigwedge_{i = 1}^n x_i' =\bigwedge_{i = 1}^n \sum_{j_i = 1}^n a_{i j_i} x_{j_i} = \notag \\ 
&= \sum_{i_1} \cdots \sum_{i_n} a_{1 j_1} \cdots a_{n j_n} x_{j_1} \wedge \cdots \wedge x_{j_n} = \notag \\ 
&=\sum_{i_1 = 1}^n \cdots \sum_{i_n = 1}^n a_{1 j_1} \cdots a_{n j_n}
\epsilon_{j_1, \dotsc, j_n} x_1 \wedge \cdots \wedge x_n = \notag \\ 
&=(\det a) \bigwedge_{i=1}^n x_i.
\end{align}

\begin{lemma}
$\frac{1}{D_{1}\cdots D_{n}}$
is equal to 
$\sum_{i=1}^{n}\frac{1}{D_{i}}\frac{1}{\prod_{\substack{j=1\\
j\ne i
}
}^{n}D_{ji}}$, where $D_{ji}=D_{j}-D_{i}$.
\label{lem:wlm1a}
\end{lemma}
\begin{proof}
\rm{
We prove the result by induction on $n$. 
For $n= 2$ or $3$, the result is evident. 
Let us now assume that this relation is true for $n=m$, with $m\geq3$,
and corresponds to our induction hypothesis.
Then, by working out the product of $(m+1)$ denominators and keeping in mind 
the induction hypothesis, 
}
\begin{align}
\frac{1}{D_{1}\hdots D_{m}D_{m+1}}&=\frac{1}{D_{m+1}}\sum_{i=1}^{m}\frac{1}{D_{i}}\frac{1}{\prod_{\substack{j=1\\
j\ne i
}
}^{m}D_{ji}}
\notag\\
&=\sum_{i=1}^{m+1}\left(\frac{1}{D_{i}}-\frac{1}{D_{m+1}}\right)\frac{1}{\prod_{\substack{j=1\\
j\ne i
}
}^{m+1}D_{ji}}\,,
\label{eq:wprof1a}
\end{align}
where we took into account the relation for $n=2$  
in the product of denominators $D_i$ and $D_{m+1}$,
we recalled that $D_{i,m+1}=-D_{m+1,i}$, and added 
a vanishing term in the sum.

In~\eqref{eq:wprof1a}, we observe that the first term corresponds to the relation for $(m+1)$
in the product of denominators $D_1,D_2,\hdots,D_m,D_{m+1}$.
This means that we are left to prove,
\begin{align}
\sum_{i=1}^{m+1}\frac{1}{\prod_{\substack{j=1\\
j\ne i
}
}^{m+1}D_{ji}}&=0\,.
\label{eq:wprof1b}
\end{align}
Because of the induction hypothesis, we can write~\eqref{eq:wprof1b} as, 
\begin{align}
&\sum_{i=1}^{m+1}\frac{1}{\prod_{\substack{j=1\\
j\ne i
}
}^{m+1}D_{ji}}
 = \notag \\
& =\sum_{i=1}^{m}\frac{1}{\prod_{\substack{j=1\\
j\ne i
}
}^{m+1}D_{ji}}+\sum_{i=1}^{m}\frac{1}{D_{i,m+1}}\frac{1}{\prod_{\substack{j=1\\
j\ne i
}
}^{m}\left(D_{j,m+1}-D_{i,m+1}\right)}
\notag\\
& = \sum_{i=1}^{m}\frac{1}{\prod_{\substack{j=1\\
j\ne i
}
}^{m+1}D_{ji}}-\sum_{i=1}^{m}\frac{1}{\prod_{\substack{j=1\\
j\ne i
}
}^{m+1}D_{ji}} = 0\,.
\end{align}
In the last line, we took into account that $D_{j,m+1}-D_{i,m+1} = D_{j,i}$ and 
$D_{i,m+1} = - D_{m+1,i}$. 
Therefore, the proof is complete by induction.
\end{proof}

\begin{lemma}
The leading singularity of the integrand 
$\int\prod_{i=1}^{n}da_{i}\,1/\left(K_{1}K_{2}\cdots K_n\right)$
is equal to $1/\det\left(x_{ij}\right)$,
where $K_i$ is a linear function in the integration variables, 
$K_i = x_{i0}+\sum_{j=1}^{n}x_{ij}\,a_j$.
\label{lem:wlm1b}
\end{lemma}
\begin{proof}
Let us calculate a composite Leray residue of degree $n$,
\begin{equation}
res^{n} = \frac{da_{1} \wedge \hdots \wedge da_{n}}{dK_{1} \wedge \hdots \wedge dK_{n}}\,.
\end{equation}
For any $i$ we have,
\begin{equation}
dK_{i} = x_{i1}da_{1} + \hdots + x_{in}da_{n}\,.
\end{equation}
Thus the residue is equal to,
\begin{equation}
\begin{split}
&res^{n} = \\
&\frac{da_{1} \wedge \hdots \wedge da_{n}}{(x_{11}da_{1} + \hdots + x_{1n}da_{n})\wedge \hdots \wedge (x_{n1}da_{1} + \hdots + x_{nn}da_{n})}\,. 
\end{split}
\end{equation}
By Lemma~\ref{wedge_det} this is equal to,
\begin{equation}
res^{n} = \frac{da_{1} \wedge \hdots \wedge da_{n}}{\det(x_{ij})da_{1} \wedge \hdots \wedge da_{n}} = \frac{1}{\det(x_{ij})}\,.
\end{equation}
Thus, proving our claim. 
\end{proof}

\section{Additional examples of one-loop Landau singularities}
\label{app:landexamples}

In this appendix, we supplement the results provided in Secs.~\ref{sec:landau}
and~\ref{sec:leading} with additional examples of explicit calculations 
of singularities of up-to six-point one-loop Feynman integrals,
and summarise the identities employed in Sec.~\ref{sec:leading} 
to find $d\log$ representations of one-loop bubble and triangle integrals 
in $D\geq n$ space-time dimension. 

To explicitly display the results of scalar one-loop box, pentagon, 
and hexagon Feynman integrals, we consider 
massless external momenta, $p_i^2=0$, and equal internal masses. 
Lastly, since Landau and leading singularities are related through theorems~\ref{th:lan1} and~\ref{th:lan2}
(or equivalently Eqs.~\eqref{eq:myleadex}), 
we only focus on the derivation of Landau singularities, by listing 
explicit expressions for the determinants of propagators and external momenta.

\paragraph{One-loop box integral\\}

By taking into account the parametrization of the loop momentum 
in terms of three independent external momenta, say $p_1,p_2, p_3$, 
and keeping generic the internal masses of propagators, 
the loop momentum components become, 
\begin{align}
a_{1}&=-\frac{m_{13}^{2}t+s\left(m_{23}^{2}-m_{34}^{2}+2s+t\right)}{2s(s+t)}\,,
\notag\\
a_{2}&=-\frac{1}{2}\left(\frac{m_{13}^{2}}{s}+\frac{m_{24}^{2}+t}{t}\right)\,,
\notag\\
a_{3}&=\frac{\left(m_{12}^{2}-m_{3}^{2}\right)t-s\left(m_{24}^{2}+t\right)}{2t(s+t)}\,,
\notag\\
\lambda_{11}&=m_{1}^{2}+a_{1}\left(a_{3}(s+t)-a_{2}s\right)-a_{2}a_{3}t\,,
\end{align}
with the Landau singularity, 
\begin{align}
&\text{LanS}_{4}^{(1)}
= \notag \\
&-\frac{1}{16}[s^{2}t^{2}-t^{2}\lambda_{\text{K}}\left(m_{1}^{2},m_{3}^{2},s\right)-s^{2}\lambda_{\text{K}}\left(m_{2}^{2},m_{4}^{2},t\right)+\notag \\
&+2st\left(m_{12}^{2}m_{34}^{2}-m_{14}^{2}m_{23}^{2}\right)]\,,
\end{align}
in terms of the kinematic invariants $s=(p_1+p_2)^2$ and $t=(p_2+p_3)^2$, 
and, to simplify the notation, we defined $m_{ij}^2\equiv m_i^2-m_j^2$. 

Thus, the Landau singularity in the equal-mass case configuration becomes, 
\begin{align}
\text{LanS}_{4}^{(1)}\Big|_{m_i^2 = m^2}
&=
\frac{s^{2}t^{2}}{16}-\frac{1}{4}m^{2}st(s+t)\,,
\end{align}
where we recall that the Gram determinant for this kinematic configuration 
is, $\det\left(p_i\cdot p_j\right) = -1/4\, s t (s+t)$. 

\paragraph{One-loop pentagon integral\\}

For the one-loop pentagon, we find for $\lambda_{11}$,
\begin{align}
\lambda_{11}
&=
m_{1}^{2}-a_{3}a_{4}s_{34}+a_{2}\left(a_{4}\left(s_{23}+s_{34}-s_{51}\right)-a_{3}s_{23}\right)
\notag\\
&+a_{1}(-a_{2}s_{12}+a_{3}\left(s_{12}+s_{23}-s_{45}\right) +\\
&+a_{4}\left(-s_{23}+s_{45}+s_{51}\right))\,,
\end{align}
in terms of five kinematic invariants, $s_{ij}=(p_i+p_j)^2$, whose
explicit expressions for $a_i=\tilde{a}_i/(16\det(p_i\cdot p_j))$ in the equal mass case are, 
\begin{align}
&\tilde{a}_{1}= -s_{12}^{2}\left(s_{23}-s_{51}\right)^{2}+ \\
&s_{12}\big(2s_{45}s_{51}^{2}-\left(s_{23}s_{34}+2\left(s_{23}+s_{34}\right)s_{45}\right)s_{51} + \notag \\
&+s_{23}s_{34}\left(s_{23}-2s_{45}\right)\big)
\notag\\
&+s_{45}\left(s_{34}-s_{51}\right)\left(s_{23}s_{34}+s_{45}\left(s_{51}-s_{34}\right)\right)\,,
\notag\\
\tilde{a}_{2}&=-s_{12}^{2}\left(s_{23}-s_{51}\right)^{2} + \notag \\
&+s_{12}\big(2s_{45}s_{51}^{2}-\left(s_{34}s_{45}+s_{23}\left(s_{34}+2s_{45}\right)\right)s_{51} + \notag \\
&+s_{23}s_{34}\left(s_{23}-s_{45}\right)\big)
\notag\\
&-s_{45}s_{51}\left(s_{23}s_{34}+s_{45}\left(s_{51}-s_{34}\right)\right)\,,
\notag\\
\tilde{a}_{3}&=-s_{12}\Big(\left(s_{34}s_{45}+\left(s_{34}+s_{45}\right)s_{51}\right)s_{23}+s_{12}\left(s_{23}-s_{51}\right)^{2} \notag \\
&-s_{23}^{2}s_{34}+s_{45}\left(s_{34}-s_{51}\right)s_{51}\Big)\,,
\notag\\
\tilde{a}_{4}&=-s_{12}s_{23}\left(-s_{23}s_{34}+s_{12}\left(s_{23}-s_{51}\right)+s_{45}\left(s_{34}+s_{51}\right)\right)\,.
\end{align}
Thus, the Landau singularity takes the form, 
\begin{align}
\text{LanS}_{5}^{(1)}\Big|_{m_i^2 = m^2}
&=
-\frac{1}{16}s_{12}s_{23}s_{34}s_{45}s_{51}+m^{2}\det\left(p_{i}\cdot p_{j}\right)\,,
\end{align}
with, 
\begin{align}
&\det\left(p_{i}\cdot p_{j}\right)=\notag \\
&\frac{1}{16}\Big[s_{12}^{2}\left(s_{23}-s_{51}\right){}^{2}+\left(s_{23}s_{34}+s_{45}\left(s_{51}-s_{34}\right)\right){}^{2}
\notag\\
&+2s_{12}\big(\left(s_{34}s_{45}+\left(s_{34}+s_{45}\right)s_{51}\right)s_{23}-s_{23}^{2}s_{34} + \notag \\
&+s_{45}\left(s_{34}-s_{51}\right)s_{51}\big)\Big]\,.
\end{align}

\paragraph{One-loop hexagon integral\\}

In order to describe the Landau singularities for the one-loop hexagon, 
we choose, without loss of generality, the following set of nine kinematic scales, 
\begin{align}
\left\{s_{12},s_{23},s_{34},s_{45},s_{56},s_{61},s_{123},s_{234},s_{345}\right\}\,,
\label{eq:6ptkin}
\end{align}
with $s_{ijk} = s_{ij}+s_{ik}+s_{jk}$. 
This choice of variables allows us to express $\lambda_{11}$ as, 
\begin{align}
&\lambda_{11}=
m_{1}^{2}+a_{1}\big(-a_{2}s_{12}+a_{3}\left(s_{12}+s_{23}-s_{123}\right) + \notag\\
&+a_{5}\left(s_{56}+s_{61}-s_{234}\right)+a_{4}\left(-s_{23}-s_{56}+s_{123}+s_{234}\right)\big) 
\notag\\
&+a_{2}\big(-a_{3}s_{23}+a_{4}\left(s_{23}+s_{34}-s_{234}\right) + \notag \\
&+a_{5}\left(-s_{34}-s_{61}+s_{234}+s_{345}\right)\big) +
\notag\\
&+a_{3}\left(a_{5}\left(s_{34}+s_{45}-s_{345}\right)-a_{4}s_{34}\right)
-a_{4}a_{5}s_{45}\,,
\end{align}
and, therefore, the Landau singularity in the equal-mass case, 
\begin{align}
&\text{LanS}_{6}^{(1)}\Big|_{m_i^2 = m^2}
= \\
&\frac{1}{64}\Big[-s_{12}^{2}s_{45}^{2}s_{234}^{2}-\left(s_{34}s_{61}s_{123}+\left(s_{23}s_{56}-s_{123}s_{234}\right)s_{345}\right){}^{2}
\notag\\
&+2s_{12}s_{45}\big(s_{123}s_{234}\left(s_{234}s_{345}-s_{34}s_{61}\right)+ \notag \\
&+s_{23}s_{56}\left(2s_{34}s_{61}-s_{234}s_{345}\right)\big)\Big] +m^{2}\det\left(p_{i}\cdot p_{j}\right)\,,
\end{align}
with the Gram determinant calculated in terms of the kinematic scales~\eqref{eq:6ptkin}. 

\paragraph{Identities used in the derivation of $d\log$ bubble and triangle Feynman integrands\\}

The results reported in Eqs.~\eqref{eq:bubbleD} and~\eqref{eq:triangleD} were obtained 
through the use of the following identities, 
\begin{subequations}
\begin{align}
&\frac{1}{D_{1}D_{2}}=\frac{1}{D_{1}}\frac{1}{D_{2}-D_{1}}-\frac{1}{D_{2}}\frac{1}{D_{2}-D_{1}}\,,\\
&\frac{dx}{\sqrt{x}\left(x+y\right)}=\frac{1}{\sqrt{-y}}d\log\frac{\sqrt{x}-\sqrt{-y}}{\sqrt{x}+\sqrt{-y}}\,, \label{eq:relDlog2} \\
&\frac{dx}{x\sqrt{\left(x+r_{1}\right)\left(x+r_{2}\right)}}= \notag \\
&= \frac{1}{\sqrt{r_{1}r_{2}}}d\log\frac{x+\sqrt{x+r_{1}}\sqrt{x+r_{2}}-\sqrt{r_{1}}\sqrt{r_{2}}}{x+\sqrt{x+r_{1}}\sqrt{x+r_{2}}+\sqrt{r_{1}}\sqrt{r_{2}}} \label{eq:relDlog3}\,,
\end{align}
\label{eq:relDlog}
\end{subequations}
where in \eqref{eq:relDlog2} and \eqref{eq:relDlog3} forms depend only on the $x$ variable while $y$, $r_{1}$ and $r_{2}$ are kept fixed.

\section{Proof of Theorem~\ref{th_lan_lead} in momentum representation}
\label{app:alt_proof}

\begin{prop}
The leading singularity of an $n$-point one-loop Feynman integral in $D=n$ space-time dimensions is equal to $ \pm 1/ \left(2^{n}\sqrt{(-1)^{D-1}\rm{LanS}}\right)$ 
\end{prop}
\begin{proof}
The one-loop Feynman integrand $\omega$ has the following form,
\begin{equation}
\omega = \frac{d^{n}k}{D_{1} \hdots D_{n}} = \frac{dk_{1} \wedge \hdots \wedge dk_{n}}{D_{1}\hdots D_{n}}\,,    
\end{equation}
where $k=(k_{1}, \hdots , k_{n})$ is an $n$-dimensional loop momentum vector and $D_{i}=(k+\sum_{j=1}^{i-1}p_{j})^{2} - m_{i}^{2}$ with $p_{j} = (p_{j1}, \hdots ,p_{jn})$ being $n$-dimensional external momenta. From the definition, the Leading singularity is the highest co-dimensional residue of the Feynman integral. Thus, let us calculate composite Leray residue of degree $n$ of the $\omega$,
\begin{equation}
res^{n}[\omega] = \frac{dk_{1} \wedge \hdots \wedge dk_{n}}{dD_{1} \wedge \hdots\wedge dD_{n}}\,.
\end{equation}
Each $dD_{i}$ is given by,
\begin{equation}
\begin{split}
dD_{i} = 2&(k_{1} + \sum_{j=1}^{i-1}p_{j1})dk_{1} + (-2)(k_{1} + \sum_{j=1}^{i-1}p_{j2})dk_{2} + \hdots + \\
&+ (-2)(k_{n} + \sum_{j=1}^{i-1}p_{jn})dk_{n}\,.
\end{split}
\end{equation}
Hence, we get,
\begin{widetext}
\begin{align}
res^{n}[\omega] = \frac{dk_{1} \wedge \hdots \wedge dk_{n}}{\Big(2k_{1}dk_{1} + \hdots +(-2)k_{n}dk_{n}\Big) \wedge \hdots \wedge \Big(2(k_{1} + \sum_{j=1}^{n-1}p_{j1})dk_{1} + \hdots + (-2)(k_{n} + \sum_{j=1}^{n-1}p_{jn})dk_{n}\Big)}\,.
\end{align}
\end{widetext}
By the Lemma~\ref{wedge_det} this can be written as,
\begin{equation}
res^{n}[\omega] = \frac{dk_{1} \wedge \hdots \wedge dk_{n}}{\det(A)dk_{1} \wedge \hdots \wedge dk_{n}} = \frac{1}{\det(A)}\,,
\end{equation}
where, 
\begin{widetext}
\begin{equation}
A=
\left(
\begin{array}{cccc}
2k_{1} & -2k_{2} & \hdots & -2k_{n} \\
2(k_{1} + p_{11}) & -2(k_{2} + p_{12}) & \hdots & -2(k_{n} + p_{1n}) \\
\vdots & \ddots & \vdots \\
2(k_{1} + \sum_{j=1}^{n-1}p_{j1}) & -2(k_{2} + \sum_{j=1}^{n-1}p_{j2}) & \hdots & -2(k_{n} + \sum_{j=1}^{n-1}p_{jn})
\end{array}
\right)\,.
\end{equation}
\end{widetext}
By the linearity of determinants, $\det(A)$ can be written as, 
\begin{equation}
\begin{split}
&\det(A) = \\ 
&\begin{vmatrix}
2k_{1} & \hdots & -2k_{n} \\
2k_{1} & \hdots & -2k_{n}  \\
\vdots & \ddots & \vdots \\
2(k_{1} + \sum_{j=1}^{n-1}p_{j1}) & \hdots & -2(k_{n} + \sum_{j=1}^{n-1}p_{jn})
\end{vmatrix}
+ \\
&\begin{vmatrix}
2k_{1} & \hdots & -2k_{n} \\
2 p_{11} & \hdots & -2p_{1n} \\
\vdots & \ddots & \vdots \\
2(k_{1} + \sum_{j=1}^{n-1}p_{j1}) & \hdots & -2(k_{n} + \sum_{j=1}^{n-1}p_{jn})
\end{vmatrix}
\,,
\end{split}
\end{equation}
where the first term is equal to zero and by repeating this procedure we end up with,
\begin{align}
&\det(A) = 
\begin{vmatrix}
2k_{1} & \hdots & -2k_{n} \\
2 p_{11} & \hdots & -2 p_{1n} \\
\vdots & \ddots & \vdots \\
2 p_{(n-1)1} & \hdots & -2p_{(n-1)n}
\end{vmatrix} = \notag \\
&= (-1)^{n-1}2^{n}
\begin{vmatrix}
k_{1} & \hdots & k_{n} \\
p_{11} & \hdots &  p_{1n} \\
\vdots & \ddots & \vdots \\
 p_{(n-1)1} & \hdots & p_{(n-1)n}
\end{vmatrix}
\,.
\end{align}
Let us denote by $B$ the following matrix,
\begin{equation}
B = 
\left(
\begin{array}{ccc}
2k_{1} & \hdots & -2k_{n} \\
2p_{11} & \hdots &  -2p_{1n} \\
\vdots & \ddots & \vdots \\
2p_{(n-1)1} & \hdots & -2p_{(n-1)n},
\end{array}
\right)\,,
\end{equation}
and we have $\det(B)=\det(A)$.
On the other hand, we have,
\begin{equation}
\begin{split}
&\det(BgB^{T}) = \det(B)\det(g)\det(B^{T}) = \\
&=\det(B)^{2} \det(g) \Rightarrow \det(B) = \pm \sqrt{\det(BgB^{T})/\det(g)}
\, ,
\end{split}
\end{equation}
where $g=diag(1,-1,\hdots,-1)$ is a metric tensor and $\det(g) = (-1)^{n-1}$. 
\\
However, $BgB^{T}$ is equal to,
\begin{equation}
BgB^{T} = 2^{2}
\left(
\begin{array}{cccc}
k \cdot k & k\cdot p_{1} & \hdots &  k\cdot p_{n-1} \\
k \cdot p_{1} & p_{1}\cdot p_{1} & \hdots & p_{1}\cdot p_{n-1} \\
\vdots & \vdots & \ddots & \vdots \\
k \cdot p_{n-1} & p_{1} \cdot p_{n-1} & \hdots & p_{n-1}\cdot p_{n-1}
\end{array}
\right).
\end{equation}
Let us express this matrix in terms of $q_{i} = k + \sum_{j=1}^{i-1}p_{j}$ vectors,
\begin{equation}
\begin{split}
k \cdot k &= q_{1} \cdot q_{1}\,, \\   
k \cdot p_{i} &= q_{1} \cdot (q_{i+1}-q_{i}) = q_{1} \cdot q_{i+1} - q_{1} \cdot q_{i}\,, \\
p_{i} \cdot p_{j} &= (q_{i+1} - q_{i})\cdot (q_{j+1} - q_{j}) = \\ 
&= q_{i+1}\cdot q_{j+1} + q_{i}\cdot q_{j} - q_{i+1}\cdot q_{j} - q_{i}\cdot q_{j+1}\,.
\end{split}
\end{equation}
This gives us
\begin{widetext}
\begin{equation}
\begin{split}
&BgB^{T} = \\
& 2^{2}
\left(
\begin{array}{cccc}
q_{1}\cdot q_{1} & q_{1} \cdot q_{2} - q_{1} \cdot q_{1} & \hdots & q_{1} \cdot q_{n} - q_{1} \cdot q_{n-1} \\
q_{1} \cdot q_{2} - q_{1} \cdot q_{1} & q_{2}\cdot q_{2} + q_{1}\cdot q_{1} - 2q_{1}\cdot q_{2} & \hdots &  q_{2}\cdot q_{n} + q_{1}\cdot q_{n-1} - q_{2}\cdot q_{n-1} - q_{1}\cdot q_{n} \\
\vdots & \vdots & \ddots & \vdots \\
q_{1} \cdot q_{n} - q_{1} \cdot q_{n-1} & q_{2}\cdot q_{n} + q_{1}\cdot q_{n-1} - q_{2}\cdot q_{n-1} - q_{1}\cdot q_{n} & \hdots & q_{n}\cdot q_{n} + q_{n-1}\cdot q_{n-1} - 2q_{n-1}\cdot q_{n}
\end{array}
\right).   
\end{split}
\end{equation}
\end{widetext}
Let us now transform $\det(BgB^{T})$, firstly by performing operations on rows. Let us add the first row to the second row and then a new second row from the third row and repeat this procedure until we reach the n-th row. After this, we will have,
\begin{equation}
\begin{split}
&\det(BgB^{T}) = \\
&=2^{2n}\begin{vmatrix}
q_{1}\cdot q_{1} & q_{1} \cdot q_{2} - q_{1} \cdot q_{1} & \hdots & q_{1} \cdot q_{n} - q_{1} \cdot q_{n-1} \\
q_{1} \cdot q_{2} & q_{2}\cdot q_{2} - q_{1}\cdot q_{2} & \hdots &  q_{2}\cdot q_{n} - q_{2}\cdot q_{n-1}\\
\vdots & \vdots & \ddots & \vdots \\
q_{1} \cdot q_{n} & q_{2}\cdot q_{n} - q_{1}\cdot q_{n} & \hdots & q_{n}\cdot q_{n}  - q_{n-1}\cdot q_{n}
\end{vmatrix}
\,.
\end{split}
\end{equation}
Let us now do the same procedure but with respect to columns,
\begin{equation}
\begin{split}
&\det(BgB^{T}) = \\
&=2^{2n}\begin{vmatrix}
q_{1}\cdot q_{1} & q_{1} \cdot q_{2}  & \hdots & q_{1} \cdot q_{n}  \\
q_{1} \cdot q_{2} & q_{2}\cdot q_{2}  & \hdots &  q_{2}\cdot q_{n} \\
\vdots & \vdots & \ddots & \vdots \\
q_{1} \cdot q_{n} & q_{2}\cdot q_{n} & \hdots & q_{n}\cdot q_{n} 
\end{vmatrix}=2^{2n}\det(q_{i}\cdot q_{j})
\,.
\end{split}
\end{equation}
This is exactly the function which together with on-shell conditions provides the Landau variety, but on-shell conditions are satisfied as we calculated the residue around polar sets $D_{1}= \hdots =D_{n}=0$. Combining together all the results we have,
\begin{equation}
\begin{split}
&\det{A} = \det(B) = \pm \sqrt{\det(BgB^{T})/\det(g)} = \\
&= \pm \sqrt{2^{2n}(-1)^{n-1}\det(q_{i} \cdot q_{j})} = \\
&= \pm 2^{n}\sqrt{(-1)^{n-1}\det(q_{i} \cdot q_{j})}
\,,    
\end{split}
\end{equation}
and finally
\begin{equation}
res^{n}[\omega] = \frac{1}{\det(A)} = \frac{\pm 1}{2^{n}\sqrt{(-1)^{n-1}\det(q_{i} \cdot q_{j})}}
\,,    
\end{equation}
that ends our proof. 
\end{proof} 

\noindent \textbf{Data Availability Statement:} \\
No Data associated in the manuscript

\vspace{50pt}
\bibliographystyle{elsarticle-num}
\bibliography{refs}

\end{document}